\documentstyle[twoside,epsf]{article}

\catcode`\@=11
\long\def\@makefntext#1{
\protect\noindent \hbox to 3.2pt {\hskip-.9pt
$^{{\eightrm\@thefnmark}}$\hfil}#1\hfill}		%CAN BE USED

\def\@makefnmark{\hbox to 0pt{$^{\@thefnmark}$\hss}}	%ORIGINAL
	
\def\ps@myheadings{\let\@mkboth\@gobbletwo
\def\@oddhead{\hbox{}
\rightmark\hfil\eightrm\thepage}
\def\@oddfoot{}\def\@evenhead{\eightrm\thepage\hfil
\leftmark\hbox{}}\def\@evenfoot{}
\def\sectionmark##1{}\def\subsectionmark##1{}}

\oddsidemargin=\evensidemargin
\newcounter{sectionc}\newcounter{subsectionc}\newcounter{subsubsectionc}
\renewcommand{\section}[1] {\vspace{12pt}\addtocounter{sectionc}{1}
\setcounter{subsectionc}{0}\setcounter{subsubsectionc}{0}\noindent
	{\tenbf\thesectionc. #1}\par\vspace{5pt}}
\renewcommand{\subsection}[1] {\vspace{12pt}\addtocounter{subsectionc}{1}
\setcounter{subsubsectionc}{0}\noindent
{\bf\thesectionc.\thesubsectionc. {\kern1pt \bfit #1}}\par\vspace{5pt}}
\renewcommand{\subsubsection}[1] {\vspace{12pt}\addtocounter{subsubsectionc}{1}
	\noindent{\tenrm\thesectionc.\thesubsectionc.\thesubsubsectionc.
	{\kern1pt \tenit #1}}\par\vspace{5pt}}
\newcommand{\nonumsection}[1] {\vspace{12pt}\noindent{\tenbf #1}
	\par\vspace{5pt}}

\newcounter{appendixc}
\newcounter{subappendixc}[appendixc]
\newcounter{subsubappendixc}[subappendixc]
\renewcommand{\thesubappendixc}{\Alph{appendixc}.\arabic{subappendixc}}
\renewcommand{\thesubsubappendixc}
	{\Alph{appendixc}.\arabic{subappendixc}.\arabic{subsubappendixc}}

\renewcommand{\appendix}[1] {\vspace{12pt}
        \refstepcounter{appendixc}
        \setcounter{figure}{0}
        \setcounter{table}{0}
        \setcounter{lemma}{0}
        \setcounter{theorem}{0}
        \setcounter{corollary}{0}
        \setcounter{definition}{0}
        \setcounter{equation}{0}
        \renewcommand{\thefigure}{\Alph{appendixc}.\arabic{figure}}
        \renewcommand{\thetable}{\Alph{appendixc}.\arabic{table}}
        \renewcommand{\theappendixc}{\Alph{appendixc}}
        \renewcommand{\thelemma}{\Alph{appendixc}.\arabic{lemma}}
        \renewcommand{\thetheorem}{\Alph{appendixc}.\arabic{theorem}}
        \renewcommand{\thedefinition}{\Alph{appendixc}.\arabic{definition}}
        \renewcommand{\thecorollary}{\Alph{appendixc}.\arabic{corollary}}
        \renewcommand{\theequation}{\Alph{appendixc}.\arabic{equation}}
        \noindent{\tenbf Appendix \theappendixc #1}\par\vspace{5pt}}
\newcommand{\subappendix}[1] {\vspace{12pt}
        \refstepcounter{subappendixc}
        \noindent{\bf Appendix \thesubappendixc. {\kern1pt \bfit #1}}
	\par\vspace{5pt}}
\newcommand{\subsubappendix}[1] {\vspace{12pt}
        \refstepcounter{subsubappendixc}
        \noindent{\rm Appendix \thesubsubappendixc. {\kern1pt \tenit #1}}
	\par\vspace{5pt}}

\newcommand{\textlineskip}{\baselineskip=13pt}
\newcommand{\smalllineskip}{\baselineskip=10pt}

\def\abstracts#1#2#3{{
	\centering{\begin{minipage}{4.5in}\footnotesize\baselineskip=10pt
	\parindent=0pt #1\par
	\parindent=15pt #2\par
	\parindent=15pt #3
	\end{minipage}}\par}}

\def\keywords#1{{
	\centering{\begin{minipage}{4.5in}\footnotesize\baselineskip=10pt
	{\footnotesize\it Keywords}\/: #1
	 \end{minipage}}\par}}
\renewenvironment{thebibliography}[1]
        {\frenchspacing
	 \ninerm\baselineskip=11pt
         \begin{list}{\arabic{enumi}.}
        {\usecounter{enumi}\setlength{\parsep}{0pt}
	 \setlength{\leftmargin 12.7pt}{\rightmargin 0pt}%FOR 1--9 ITEMS
         \setlength{\itemsep}{0pt} \settowidth
	{\labelwidth}{#1.}\sloppy}}{\end{list}}

\newcounter{itemlistc}
\newcounter{romanlistc}
\newcounter{alphlistc}
\newcounter{arabiclistc}

\newcommand{\fcaption}[1]{
        \refstepcounter{figure}
        \setbox\@tempboxa = \hbox{\footnotesize Fig.~\thefigure. #1}
        \ifdim \wd\@tempboxa > 5in
           {\begin{center}
        \parbox{5in}{\footnotesize\smalllineskip Fig.~\thefigure. #1}
            \end{center}}
        \else
             {\begin{center}
             {\footnotesize Fig.~\thefigure. #1}
              \end{center}}
        \fi}

\newcommand{\tcaption}[1]{
        \refstepcounter{table}
        \setbox\@tempboxa = \hbox{\footnotesize Table~\thetable. #1}
        \ifdim \wd\@tempboxa > 5in
           {\begin{center}
        \parbox{5in}{\footnotesize\smalllineskip Table~\thetable. #1}
            \end{center}}
        \else
             {\begin{center}
             {\footnotesize Table~\thetable. #1}
              \end{center}}
        \fi}

\def\pmb#1{\setbox0=\hbox{#1}
	\kern-.025em\copy0\kern-\wd0
	\kern.05em\copy0\kern-\wd0
	\kern-.025em\raise.0433em\box0}

\def\fnt#1#2{\footnotetext{\kern-.3em
	{$^{\mbox{\scriptsize #1}}$}{#2}}}

\def\fpage#1{\begingroup
\voffset=.3in
\thispagestyle{empty}\begin{table}[b]\centerline{\footnotesize #1}
	\end{table}\endgroup}

\def\runninghead#1#2{\pagestyle{myheadings}
\markboth{{\protect\footnotesize\it{\quad #1}}\hfill}
{\hfill{\protect\footnotesize\it{#2\quad}}}}
\font\tenrm=cmr10
\font\tenit=cmti10
\font\tenbf=cmbx10
\font\bfit=cmbxti10 at 10pt
\font\ninerm=cmr9

\font\eightrm=cmr8

\newtheorem{theorem}{\indent Theorem}
\newtheorem{lemma}{Lemma}
\newtheorem{corollary}{Corollary}
\newtheorem{proposition}{Proposition}
\newtheorem{definition}{Definition}
\newtheorem{remark}{Remark}
\newtheorem{defi/prop}{Definition/Proposition}

\def\FigName{figure}%
\newbox\captionbox
\long\def\@makecaption#1#2{%
  \ifx\FigName\@captype
    \vskip\abovecaptionskip
    \setbox\tempbox\hbox{{\figurecaptionfont #1\hskip1em #2}}
	\ifdim\wd\tempbox< 28pc
	\centerline{\box\tempbox}
	\else
	{\figurecaptionfont #1\hskip1em #2\par}
\fi\else
  	\setbox\tempbox\hbox{{\tablecaptionfont #1\hskip1em #2}}
 	\ifdim\wd\tempbox< 28pc
	\centerline{\box\tempbox}
	\else
	{\tablecaptionfont #1\hskip1em #2\par}%
	\fi
 \vskip\belowcaptionskip
 \fi}
\InputIfFileExists{psfig.sty}
{\typeout{^^Jpsfig.sty inputed...ok}}{\typeout{^^JWarning: psfig.sty could be be found.^^J}}
\InputIfFileExists{epsfsafe.tex}
{\typeout{^^Jepsfsafe.tex inputed...ok}}
			{\typeout{^^JWarning: epsfsafe.tex could not be found.^^J}}
\InputIfFileExists{epsfig.sty}
{\typeout{^^Jepsfig.sty inputed...ok}}{\typeout{^^JWarning: epsfig.sty could not be found.^^J}}
\InputIfFileExists{epsf.sty}
{\typeout{^^Jepsf.sty inputed...ok}}{\typeout{^^JWarning: epsf.sty could not be found.^^J}}%
\def\fps@figure{tbp}
\def\ftype@figure{1}
\def\ext@figure{lof}
\def\fnum@figure{Fig.\ \thefigure}
\def\qed{\hbox{${\vcenter{\vbox{	          %HOLLOW SQUARE
   \hrule height 0.4pt\hbox{\vrule width 0.4pt height 6pt
   \kern5pt\vrule width 0.4pt}\hrule height 0.4pt}}}$}}

\newcommand{\e}{\varepsilon}

\newcommand{\cH}{{\mathcal{H}}}

\newcommand{\vol}{{\mathrm{vol}}}
\newcommand{\vrad}{{\mathrm{vrad}}}
\newcommand{\conv}{{\mathrm{conv}}}

\newcommand{\card}{{\mathrm{card}}}
\newcommand{\spec}{{\mathrm{spec}}}
\newcommand{\diag}{{\mathrm{diag}}}

\newcommand{\E}{{\mathbf{E}}}

\newcommand{\scalar}[2]{\langle #1 , #2\rangle}
\newcommand{\braket}[2]{\langle #1 | #2\rangle}
\newcommand{\ketbra}[2]{| #1 \rangle \langle #2 |}
\newcommand{\bra}[1]{\langle #1 |}
\newcommand{\ket}[1]{| #1 \rangle}

\begin{document}
\setlength{\textheight}{8.0truein}    %FOR 2ND PAGE ONWARDS

\runninghead{Locally restricted measurements on a multipartite quantum system: data hiding is generic}
            {Guillaume Aubrun and C\'{e}cilia Lancien}

\normalsize\textlineskip
\thispagestyle{empty}
\setcounter{page}{1}

%\copyrightheading{}	%	{Vol.~1, No.~0 (2001) 000--000}

\vspace*{0.88truein}

\fpage{1}
\centerline{\bf LOCALLY RESTRICTED MEASUREMENTS }
\vspace*{0.035truein}
\centerline{\bf ON A MULTIPARTITE QUANTUM SYSTEM: }
\vspace*{0.035truein}
\centerline{\bf  DATA HIDING IS GENERIC}
\vspace*{0.37truein}
\centerline{\footnotesize
GUILLAUME AUBRUN\footnote{aubrun@math.univ-lyon1.fr}}
\vspace*{0.015truein}
\centerline{\footnotesize\it Institut Camille Jordan, Universit\'{e} Claude Bernard Lyon 1}
\baselineskip=10pt
\centerline{\footnotesize\it 43 boulevard du 11 novembre 1918, 69622 Villeurbanne Cedex, France}
\vspace*{10pt}
\centerline{\footnotesize C\'{E}CILIA LANCIEN\footnote{lancien@math.univ-lyon1.fr}}
\vspace*{0.015truein}
\centerline{\footnotesize\it Institut Camille Jordan, Universit\'{e} Claude Bernard Lyon 1}
\baselineskip=10pt
\centerline{\footnotesize\it 43 boulevard du 11 novembre 1918, 69622 Villeurbanne Cedex, France}
\vspace*{10pt}
\vspace*{0.015truein}
\centerline{\footnotesize\it F\'{\i}sica Te\`{o}rica: Informaci\'{o} i Fenomens Qu\`{a}ntics, Universitat Aut\`{o}noma de Barcelona}
\baselineskip=10pt
\centerline{\footnotesize\it ES-08193 Bellaterra (Barcelona), Spain}
%\vspace*{10pt}
%\vspace*{0.225truein}
%\publisher{(received date)}{(revised date)}

\vspace*{0.21truein}
\abstracts{ We study the distinguishability norms associated to families of locally restricted POVMs on multipartite systems.
These norms (introduced by Matthews, Wehner and Winter) quantify
how quantum measurements, subject to locality constraints, perform in the task of discriminating two multipartite quantum states.
We mainly address the following question regarding the behaviour of these
distinguishability norms in the high-dimensional regime: On a bipartite space, what are the relative strengths of standard classes of
locally restricted measurements? We show that the class of PPT measurements typically
performs almost as well as the class of all measurements whereas restricting to local measurements and classical communication,
or even just to separable measurements, implies a substantial loss. We also provide examples of state pairs which can be perfectly
distinguished by local measurements if (one-way) classical communication is allowed between the parties, but very poorly without it.
Finally, we study how many POVMs are needed to distinguish almost perfectly any pair of states on ${\mathbf{C}}^d$,
showing that the answer is $\exp(\Theta(d^2))$.}{}{}
\vspace*{10pt}

\vspace*{3pt}

%\communicate{to be filled by the Editorial}

%\vspace*{1pt}\textlineskip	%) USE THIS MEASUREMENT WHEN THERE IS

\section{Introduction}

How quantum measurements can help us make decisions? We consider a basic problem, the task of distinguishing two quantum states,
where this question has a neat answer. Given a POVM (Positive Operator-Valued Measure) $\mathrm{M}$ on ${\mathbf{C}}^d$, Matthews, Wehner and Winter \cite{MWW}
introduced its distinguishability norm $\|\cdot\|_{{\mathrm{M}}}$, which has the property that given a pair $(\rho,\sigma)$ of quantum states,
$\|\rho-\sigma\|_{{\mathrm{M}}}$ is the bias observed when the POVM ${\mathrm{M}}$ is used optimally to distinguish $\rho$ from $\sigma$ (the larger
is the norm, the more efficient is the POVM).
More generally, we can associate to a family of POVMs ${\mathbf{\underline{M}}}$ the norm $\|\cdot\|_{{\mathbf{\underline{M}}}} = \sup \{ \|\cdot\|_{{\mathrm{M}}} \ : \
{\mathrm{M}} \in {\mathbf{\underline{M}}} \}$ which corresponds to the bias achieved by the best POVM from the family.

In this paper, we study these norms from a functional-analytic point of view and
are mostly interested in the asymptotic regime, when the dimension of the underlying Hilbert space tends to infinity.

\subsection{How many essentially distinct POVMs are there?}

The (infinite) family ${\mathbf{\underline{ALL}}}$ of all POVMs on ${\mathbf{C}}^d$ achieves maximal efficiency in the distinguishability task, and in some sense gives us perfect information. It was indeed one of
the seminal observations by Holevo \cite{Holevo} and Helstrom \cite{Helstrom} that $\|\cdot\|_{{\mathbf{\underline{ALL}}}} = \|\cdot\|_1$, so that two orthogonal quantum states could be perfectly
distinguished (i.e. with a zero probability of error) by a suitable measurement.
But how ``complex'' is the class ${\mathbf{\underline{ALL}}}$? What about finite subfamilies? How many POVMs are needed to obtain near-to-optimal efficiency?
We show (Theorem \ref{theorem:approximation-of-ALL} in Section 2.2)
that $\exp(\Theta(d^2))$ different POVMs are necessary (and sufficient) to obtain approximation within a constant factor. The concept of mean width (from convex geometry) plays an important
role in our proof,
which is detailed in Section 3.

\subsection{Locally restricted POVMs on a multipartite quantum system}

On a multipartite quantum system, experimenters usually cannot implement any global observable. For instance, they may be only able to perform
quantum measurements on their own subsystem
(and then perhaps to communicate the results classically). A natural question in such situation is thus to quantify the relative strengths of
several classes of measurements, restricted by these locality constraints, such as LOCC,
separable or PPT measurements (precise definitions appear in Section 2.3).

Let us summarize the main result in this paper (restricting here to the bipartite case for the sake of clarity). We consider typical discrimination tasks, in the following sense.
Let $\rho$ and $\sigma$ be states chosen independently and uniformly at random within
the set of all states on ${\mathbf{C}}^d\otimes{\mathbf{C}}^d$. We show that our ability to distinguish $\rho$ from $\sigma$ depends in an essential way on the class
of the allowed measurements.
Indeed, with high probability, $\|\rho-\sigma\|_{{\mathbf{\underline{PPT}}}}$ is of order $1$ (as $\|\rho-\sigma\|_{{\mathbf{\underline{ALL}}}}$) while
$\|\rho-\sigma\|_{{\mathbf{\underline{SEP}}}}$, $\|\rho-\sigma\|_{{\mathbf{\underline{LOCC}}}}$ and $\|\rho-\sigma\|_{{\mathbf{\underline{LOCC}^\rightarrow}}}$ are of order $1/\sqrt{d}$.
This shows that data hiding is generic: typically, high-dimensional quantum states cannot be distinguished locally even
though they look different globally.

These results appear as Theorem \ref{theorem:typical-states} in Section 2.4. The proofs
are detailed in Section 5. They rely, as a first essential step, on estimates
on the volume radius and the mean width of the (polar of) the unit balls associated to the norms
$\|\cdot\|_{{\mathbf{\underline{PPT}}}}$, $\|\cdot\|_{{\mathbf{\underline{SEP}}}}$ and $\|\cdot\|_{{\mathbf{\underline{LOCC}}}}$ (Theorem \ref{theorem:vrad-w-PPT-SEP}).
We gathered tools and results from convex geometry in an Appendix. The use of concentration of measure and random matrix theory
(Proposition \ref{proposition:Delta-vs-difference-of-states}) then allows to pass from these global estimates to the estimates in a typical
direction quoted above. In Section 6 corollaries on quantum data hiding are derived and detailed, both in the bipartite
and in the generalized multipartite case.

We also provide examples of random bipartite states $\rho,\sigma$ on ${\mathbf{C}}^d\otimes{\mathbf{C}}^d$ which are such that $\|\rho-\sigma\|_{{\mathbf{\underline{LOCC}^\rightarrow}}}=2$ while, with high probability, $\|\rho-\sigma\|_{{\mathbf{\underline{LO}}}}$ is of order $1/\sqrt{d}$. The precise result appears as Theorem \ref{theorem:LO-vs-LOCC} in Section 2.4 and is proved in Section 4.

\subsection{Notation}

We denote by $\cH({\mathbf{C}}^d)$ the set of Hermitian operators on ${\mathbf{C}}^d$, and by $\cH_+({\mathbf{C}}^d)$ the subset of positive operators. We denote by $\|\cdot\|_1$ the trace class norm,
by $\|\cdot\|_{\infty}$ the operator norm
and by $\|\cdot\|_{2}$ the Hilbert--Schmidt norm. When $A,B$ are self-adjoint matrices, we denote by $[A,B]$ the order interval, i.e. the set of self-adjoint
matrices $C$ such that both $C-A$ and $B-C$ are nonnegative. In particular, $[-{\mathrm{Id}},{\mathrm{Id}}]$ is the self-adjoint part of the unit ball for $\|\cdot\|_{\infty}$.
We also denote by $\|\cdot\|_2$ the Euclidean norm on ${\mathbf{R}}^n$ or ${\mathbf{C}}^n$.

The letters $C,c,c_0,\dots$ denote numerical constants, independent from any other parameters such as the dimension. The value of these constants may change
from occurrence to occurrence. When $A$ and $B$ are quantities depending on the dimension, the notation $A \,\preceq\, B$ means that there is a constant $C$ such that
$A \leq CB$. The notation $A \simeq B$ means both $A \,\preceq\, B$ and $B \,\preceq\, A$, and $A \sim B$ means that the ratio $A/B$ tends to $1$ when the dimension tends
to infinity.

Extra notation, concepts and results from convex geometry are introduced in Appendix \ref{ap:convex-geometry}.

\section{Distinguishing quantum states: survey of our results}
\label{sec:POVM-geometry}

\subsection{General setting}
\label{sec:measurement-norm}

In this section, we gather some basic information about norms associated to POVMs, and refer to \cite{MWW} for more details and proofs.
A POVM (Positive Operator-Valued Measure) on ${\mathbf{C}}^d$ is a finite family ${\mathrm{M}}=(M_i)_{i\in I}$ of positive operators on ${\mathbf{C}}^d$ such that
\[ \sum_{i\in I} M_i = {\mathrm{Id}} . \]
One could consider also continuous POVMs, where the finite sum is replaced by an integral. However this is not necessary, since continuous POVMs appear as limit
cases of discrete POVMs which we consider here (see e.g. \cite{AL}).

Given a POVM ${\mathrm{M}}=(M_i)_{i\in I}$ on ${\mathbf{C}}^d$, and denoting by $\{|i\rangle,\ i\in I\}$ an orthonormal basis of ${\mathbf{C}}^{\card(I)}$,
we may associate to
${\mathrm{M}}$ the CPTP (Completely Positive and Trace-Preserving) map
\[ {\mathcal{M}} : \Delta \in{\mathcal{H}}({\mathbf{C}}^d)\mapsto\sum_{i\in I}\big({\mathrm{Tr}} M_i\Delta\big)|i\rangle\langle i|\in{\mathcal{H}}\big({\mathbf{C}}^{\card(I)}\big). \]
The measurement (semi-)norm associated to ${\mathrm{M}}$ is then defined for $\Delta\in{\mathcal{H}}({\mathbf{C}}^d)$ as
\[ \|\Delta\|_{{\mathrm{M}}}:=\|{{\mathcal{M}}}(\Delta)\|_1 = \sum_{i\in I} |{\mathrm{Tr}} M_i\Delta| . \]
Note that for any $\Delta\in\cH({\mathbf{C}}^d)$, $\|\Delta\|_{{\mathrm{M}}}\leq\|\Delta\|_1$, with equality if $\Delta\in\cH_+({\mathbf{C}}^d)$.

In general, $\|\cdot\|_{{\mathrm{M}}}$ is a semi-norm, and may vanish on non-zero Hermitians. A necessary and sufficient condition for $\|\cdot\|_{{\mathrm{M}}}$ to be
a norm is that the POVM ${\mathrm{M}}=(M_i)_{i \in I}$ is informationally complete, i.e. that the family of operators $(M_i)_{i\in I}$ spans ${\mathcal{H}}({\mathbf{C}}^d)$ as a linear space.
This especially implies that ${\mathrm{M}}$ has a total number of outcomes satisfying $\card (I) \geq d^2=\mathrm{dim}\ {\mathcal{H}}({\mathbf{C}}^d)$.

We denote by $B_{\|\cdot\|_{{\mathrm{M}}}}$ the unit ball associated to $\|\cdot\|_{{\mathrm{M}}}$, and by $K_{{\mathrm{M}}}$ the polar of $B_{\|\cdot\|_{{\mathrm{M}}}}$ (i.e. the
unit ball associated to the norm dual to $\|\cdot\|_{{\mathrm{M}}}$). In
other words, the support function of $K_{{\mathrm{M}}}$ is defined for $\Delta\in{\mathcal{H}}({\mathbf{C}}^d)$ as
\begin{equation} \label{eq:definition-K} h_{K_{{\mathrm{M}}}}(\Delta) = \|\Delta\|_{{\mathrm{M}}}. \end{equation}
Precise definitions of these concepts are given in Appendix \ref{ap:convex-geometry}.

More generally, one can define the ``measurement'' or ``distinguishability'' norm associated to a whole set ${\mathbf{\underline{M}}}$ of POVMs on ${\mathbf{C}}^d$ as
\[ \|\cdot\|_{{\mathbf{\underline{M}}}}:=\sup_{{\mathrm{M}}\in{\mathbf{\underline{M}}}}\|\cdot\|_{{\mathrm{M}}}. \]
The corresponding unit ball, and its polar, are
\[ B_{\|\cdot\|_{{\mathbf{\underline{M}}}}} = \bigcap_{{\mathrm{M}} \in {\mathbf{\underline{M}}}} B_{\|\cdot\|_{{\mathrm{M}}}} ,\]
\[ K_{{\mathbf{\underline{M}}}} = \conv \left( \bigcup_{{\mathrm{M}} \in {\mathbf{\underline{M}}}} K_{{\mathrm{M}}} \right) .\]

As mentioned earlier on in the Introduction, these measurement norms are related to the task of distinguishing quantum states.
Let us consider the situation where a system (with associated Hilbert space ${\mathbf{C}}^d$) can be either in state $\rho$ or in state $\sigma$, with equal prior
probabilities $\frac{1}{2}$. It is known \cite{Holevo, Helstrom} that a decision process based on the maximum likelihood rule after performing the POVM ${\mathrm{M}}$
on the system yields a probability of error
\[ {\mathbf{P}}_e  =\frac{1}{2}\left(1-\left\|\frac{1}{2} \rho-\frac{1}{2}\sigma\right\|_{{\mathrm{M}}}\right). \]
In this context, the operational interpretation of the quantity $\|\rho-\sigma\|_{{\mathrm{M}}}$ is thus clear (and actually justifies the terminology of ``distinguishability norm''): up to a factor $1/2$, it is nothing else than the bias of the POVM ${\mathrm{M}}$ on the state pair $(\rho,\sigma)$.

Something that is worth pointing out is that, for any set ${\mathbf{\underline{M}}}$ of POVMs on ${\mathbf{C}}^d$, there exists a set ${\mathbf{\widetilde{\underline{M}}}}$ of $2$-outcome POVMs on
${\mathbf{C}}^d$ which is such that $\|\cdot\|_{{\mathbf{\underline{M}}}}=\|\cdot\|_{{\mathbf{\widetilde{\underline{M}}}}}$. It may be explicitly defined as
\[ {\mathbf{\widetilde{\underline{M}}}}:=\left\{\big(M,{\mathrm{Id}}-M\big) \ : \ \exists\ (M_i)_{i\in I}\in{\mathbf{\underline{M}}},\ \exists\ \widetilde{I}\subset I:\
M=\sum_{i\in\widetilde{I}} M_i\right\} .\]
Note then that
\[ K_{{\mathbf{\underline{M}}}}=\conv \big\{2M-{\mathrm{Id}},\ (M,{\mathrm{Id}}-M)\in\mathbf{\widetilde{\underline{M}}}\big\}. \]

\subsection{On the complexity of the class of all POVMs}

Denote by ${\mathbf{\underline{ALL}}}$ the family of all POVMs on ${\mathbf{C}}^d$. As we already noticed, $\|\cdot\|_{{\mathbf{\underline{ALL}}}} = \|\cdot\|_1$ and therefore
$K_{{\mathbf{\underline{ALL}}}}$ equals $[-{\mathrm{Id}},{\mathrm{Id}}]$, which is the unit ball in $\cH({\mathbf{C}}^d)$ for the operator norm.

The family ${\mathbf{\underline{ALL}}}$ is obviously infinite. Since real-life situations can involve only finitely many apparatuses,
it makes sense to ask what must be the cardinality of
a finite family of POVMs ${\mathbf{\underline{M}}}$ which achieves close to perfect discrimination, i.e.
such that the inequality $\|\cdot\|_{{\mathbf{\underline{M}}}} \geq \lambda \|\cdot\|_{{\mathbf{\underline{ALL}}}}$ holds for some $0<\lambda<1$. We show that the answer
is exponential in $d^2$. More precisely, we have the theorem below.

\begin{theorem} \label{theorem:approximation-of-ALL}
There are positive constants $c,C$ such that the following holds
\begin{enumerate}
 \item[(i)] For any dimension $d$ and any $0<\e<1$, there is a family ${\mathbf{\underline{M}}}$ consisting of at most  $\exp(C |\log\e| d^2)$ POVMs on
 ${\mathbf{C}}^d$ such that
 $\|\cdot\|_{{\mathbf{\underline{M}}}} \geq (1-\e) \|\cdot\|_{{\mathbf{\underline{ALL}}}}$.
 \item[(ii)] For any $\e>C/\sqrt{d}$, any family ${\mathbf{\underline{M}}}$ of POVMs on ${\mathbf{C}}^d$ such that $\|\cdot\|_{{\mathbf{\underline{M}}}} \geq \e \|\cdot\|_{{\mathbf{\underline{ALL}}}}$
 contains at least $\exp(c \e^2 d^2)$ POVMs.
 \end{enumerate}
\end{theorem}

Theorem \ref{theorem:approximation-of-ALL} is proved in Section 3.
It is clear that the conclusion of (ii) fails for $\e \,\preceq\, 1/\sqrt{d}$, since a single POVM
${\mathrm{M}}$ (e.g. the uniform POVM, see \cite{MWW}) may satisfy $\|\cdot\|_{{\mathrm{M}}} \,\succeq\, \frac{1}{\sqrt{d}} \|\cdot\|_1$.

\subsection{Locally restricted measurements on a bipartite quantum system}
\label{sec:local-POVM}

We now study the class of locally restricted POVMs. We assume that the underlying global Hilbert space is the tensor product of several local
Hilbert spaces. However, for simplicity,
we focus on the case of a bipartite system
in which both parts play the same role and consider the Hilbert space ${\mathcal{H}} = {\mathbf{C}}^d \otimes {\mathbf{C}}^d$.
Several classes of POVMs can be defined on ${\mathcal{H}}$ due to various levels of locality restrictions (consult \cite{MWW} or \cite{LW} for further information).

The most restricted class of POVMs on ${\mathcal{H}}$ is the one of local measurements, whose elements are tensor products of measurements on each of the
sub-systems:
\[{\mathbf{\underline{LO}}}:=\left\{\left(M_{i} \otimes N_{j} \right)_{i\in I,j \in J} \ : \ \ M_i \geq 0, \ N_j \geq 0, \ \sum_{i \in I} M_i = {\mathrm{Id}}_{{\mathbf{C}}^d}, \ \sum_{j \in J} N_j = {\mathrm{Id}}_{{\mathbf{C}}^d} \right\}. \]
This corresponds to the situation where parties are not allowed to communicate.

Then, we consider the class of separable measurements, whose elements are the measurements on ${\mathcal{H}}$ made of tensor operators
\[ {\mathbf{\underline{SEP}}}:=\left\{\left(M_j \otimes N_j\right)_{j\in J} \ : \ \ M_j \geq 0, \ N_j \geq 0, \ \sum_{j\in J}M_j \otimes N_j ={\mathrm{Id}}_{{\mathbf{C}}^d \otimes {\mathbf{C}}^d} \right\}. \]

An important subclass of ${\mathbf{\underline{SEP}}}$ is the class ${\mathbf{\underline{LOCC}}}$ (Local Operations and Classical Communication)
of measurements that can be implemented by a finite sequence of local operations on the sub-systems
followed by classical communication between the parties. This class can be described recursively as the smallest subclass of ${\mathbf{\underline{SEP}}}$ which
contains ${\mathbf{\underline{LO}}}$ and is stable under the following operation: given a POVM ${\mathrm{M}}=(M_i)_{i \in I}$ on ${\mathbf{C}}^d$, and for each $i \in I$ a
${\mathbf{\underline{LOCC}}}$ POVM $\left(N^{(1)}_j \otimes N^{(2)}_j\right)_{j \in J_i}$, the POVMs
\[ \left( M_i^{1/2} N^{(1)}_j M_i^{1/2} \otimes N^{(2)}_j \right)_{i \in I, j \in J_i} \ \ \textnormal{and}
\ \ \left( N^{(1)}_j \otimes M_i^{1/2}  N^{(2)}_j M_i^{1/2} \right)_{i \in I, j \in J_i} \]
are in ${\mathbf{\underline{LOCC}}}$. A subclass of ${\mathbf{\underline{LOCC}}}$ is the class ${\mathbf{\underline{LOCC}^\rightarrow}}$ of one-way LOCC POVMs, which has a simpler description
\[{\mathbf{\underline{LOCC}^\rightarrow}}:=\left\{\left(M_{i} \otimes N_{i,j} \right)_{i\in I,j \in J_i} \ : \ \ M_i \geq 0, \ N_{i,j} \geq 0,
\ \sum_{i \in I} M_i = {\mathrm{Id}}_{{\mathbf{C}}^d}, \ \sum_{j \in J_i} N_{i,j} = {\mathrm{Id}}_{{\mathbf{C}}^d} \right\}. \]

Finally, we consider the class of positive under partial transpose (PPT) measurements, whose elements are the measurements on ${\mathcal{H}}$
made of operators that
remain positive when partially transposed on one sub-system:
\[ {\mathbf{\underline{PPT}}}:=\left\{(M_j)_{j\in J} \ : \ M_j \geq 0,\ M_j^{\Gamma}\geq\mathrm{0},\ \sum_{j\in J}M_j={\mathrm{Id}}_{{\mathbf{C}}^d \otimes {\mathbf{C}}^d} \right\}. \]
The partial transposition $\Gamma$ is defined by its action on tensor
operators on ${\mathcal{H}}$: $(M\otimes N)^{\Gamma}:= M^T \otimes N$, $M^T$ denoting the usual transpose of $M$.
Let us point out that, even though the expression of a matrix transpose depends on the chosen basis, its eigenvalues on the contrary are intrinsic.
Therefore the PPT notion is basis-independent.

It is clear from the definitions that we have the chain of inclusions
\[ {\mathbf{\underline{LO}}}\subset{\mathbf{\underline{LOCC}^\rightarrow}}\subset {\mathbf{\underline{LOCC}}} \subset{\mathbf{\underline{SEP}}}\subset{\mathbf{\underline{PPT}}}\subset {\mathbf{\underline{ALL}}} \]
and consequently the chain of norm inequalities
\begin{equation}\label{eq:hierarchy} \|\cdot\|_{{\mathbf{\underline{LO}}}} \leq \|\cdot\|_{{\mathbf{\underline{LOCC}^\rightarrow}}} \leq \|\cdot\|_{{\mathbf{\underline{LOCC}}}}
\leq \|\cdot\|_{ {\mathbf{\underline{SEP}}}} \leq
\|\cdot\|_{{\mathbf{\underline{PPT}}}} \leq \|\cdot\|_{{\mathbf{\underline{ALL}}}}. \end{equation}

All the inequalities in \ref{eq:hierarchy} are known to be strict provided $d >2$. Note though that the difference between the norms
$\|\cdot\|_{{\mathbf{\underline{LOCC}^\rightarrow}}}$ and $\|\cdot\|_{{\mathbf{\underline{LOCC}}}}$, as well as between $\|\cdot\|_{{\mathbf{\underline{LOCC}}}}$ and $\|\cdot\|_{{\mathbf{\underline{SEP}}}}$,
has been established only very recently (see \cite{CH}).

Here, we are interested in the high-dimensional behaviour of these norms, and the general question we investigate is whether or not the various gaps in the
hierarchy are bounded (independently of the dimension of the subsystems). It is already known that the gap between ${\mathbf{\underline{PPT}}}$ and ${\mathbf{\underline{ALL}}}$ is unbounded,
an important example being provided by the symmetric state $\varsigma$ and the antisymmetric state $\alpha$ on ${\mathbf{C}}^d\otimes{\mathbf{C}}^d$ which satisfy (see e.g. \cite{DVLT2})
\[ \| \varsigma - \alpha \|_{{\mathbf{\underline{ALL}}}} =2  \ \ \textnormal{ while } \ \ \| \varsigma - \alpha \|_{{\mathbf{\underline{PPT}}}} =
\frac{4}{d+1}.\]
We show however (see Theorem \ref{theorem:typical-states}) that such feature is not generic. This is in contrast with the gap between ${\mathbf{\underline{SEP}}}$ and
${\mathbf{\underline{PPT}}}$ which we prove to be generically unbounded (see Theorem \ref{theorem:typical-states}). We also provide examples of unbounded gap between ${\mathbf{\underline{LO}}}$
and ${\mathbf{\underline{LOCC}^\rightarrow}}$ (see Theorem \ref{theorem:LO-vs-LOCC}) but we do not know if this situation is typical. Regarding the gaps between
${\mathbf{\underline{LOCC}^\rightarrow}}$, ${\mathbf{\underline{LOCC}}}$ and ${\mathbf{\underline{SEP}}}$, determining whether they are bounded remains an open problem.

Note also that for states of low rank, the gaps between these norms remain bounded. It follows from the results of \cite{LW} that, for $\Delta \in \cH({\mathbf{C}}^d \otimes {\mathbf{C}}^d)$
of rank $r$,
we have
\[ \|\Delta\|_{\mathbf{\underline{LO}}} \geq \frac{1}{18\sqrt{r}} \|\Delta\|_{{\mathbf{\underline{ALL}}}}. \]

\subsection{Discriminating power of the different classes of locally restricted measurements}
\label{sec:local-POVM-statement}

Our main result compares the efficiency of the classes ${\mathbf{\underline{LOCC}^\rightarrow}}$, ${\mathbf{\underline{LOCC}}}$, ${\mathbf{\underline{SEP}}}$, ${\mathbf{\underline{PPT}}}$ and ${\mathbf{\underline{ALL}}}$ to perform a typical discrimination task.
Here ``typical'' means the following: we consider the problem of distinguishing $\rho$ from $\sigma$, where $\rho$ and $\sigma$ are random states, chosen independently at random with respect to the
uniform measure (i.e. the
Lebesgue measure induced by the Hilbert--Schmidt distance) on the set of all states.
It turns out that the PPT constraint on the allowed measurements is not very restrictive, affecting typically the performance by only a constant
factor, while the separability one implies
a more substantial loss.
This shows that generic bipartite states are data hiding: separable measurements (and even more so local measurements followed by classical communication) can poorly distinguish them (see \cite{HLSW} for another
instance of this phenomenon and Section 6 for a more detailed discussion on that topic).

\begin{theorem} \label{theorem:typical-states}
There are universal constants $C,c$ such that the following holds. Given a dimension $d$,
let $\rho$ and $\sigma$ be random states, independent and uniformly distributed on the set of states on ${\mathbf{C}}^d \otimes {\mathbf{C}}^d$. Then, with high probability,
\[ c \leq \| \rho - \sigma \|_{{\mathbf{\underline{PPT}}}} \leq \| \rho - \sigma \|_{{\mathbf{\underline{ALL}}}} \leq C, \]
\[ \frac{c}{\sqrt{d}} \leq \| \rho - \sigma \|_{{\mathbf{\underline{LOCC}^\rightarrow}}} \leq \| \rho - \sigma \|_{{\mathbf{\underline{LOCC}}}} \leq \| \rho - \sigma \|_{{\mathbf{\underline{SEP}}}} \leq \frac{C}{\sqrt{d}}. \]
Here, ``with high probability'' means that the probability that one of the conclusions fails is less than $\exp(-c_0 d)$ for some constant $c_0>0$.
\end{theorem}

An immediate consequence of the high probability estimates is that one can find in ${\mathbf{C}}^d \otimes {\mathbf{C}}^d$ exponentially many states which are pairwise data hiding.

\begin{corollary}
There are constants $C,c$ such that, if $\mathcal{A}$ denotes a set of $\exp(cd)$ independent random states uniformly distributed on the set
of states on ${\mathbf{C}}^d \otimes {\mathbf{C}}^d$, with high probability any pair of distinct states $\rho,\sigma \in \mathcal{A}$ satisfies the conclusions of
Theorem \ref{theorem:typical-states}.
\end{corollary}

We deduce Theorem \ref{theorem:typical-states} from estimates on the mean width and the volume of the unit balls $K_{{\mathbf{\underline{LOCC}^\rightarrow}}}$,
$K_{{\mathbf{\underline{SEP}}}}$ and $
K_{{\mathbf{\underline{PPT}}}}$. The use of concentration
of measure allows to pass from these global estimates to the estimates in a typical direction that appear in Theorem \ref{theorem:typical-states}.
We include all this material
in Section 5.

We also show that even the smallest amount of communication has a huge influence:
we give examples of states which are perfectly distinguishable under local measurements and one-way classical communication
but very poorly distinguishable under local measurements with no communication between the parties.

\begin{theorem} \label{theorem:LO-vs-LOCC}
There is a universal constant $C$ such that the following holds: for any dimension $d$, there exists states $\rho$ and $\sigma$ on ${\mathbf{C}}^d \otimes {\mathbf{C}}^d$
such that
\[ \|\rho-\sigma\|_{{\mathbf{\underline{LOCC}^\rightarrow}}} = 2, \]
and
\begin{equation} \label{equation:small-LO} \|\rho-\sigma\|_{{\mathbf{\underline{LO}}}} \leq \frac{C}{\sqrt{d}}. \end{equation}
\end{theorem}

These states are constructed as follows: assuming without loss of generality that $d$ is even, let
$E$ be a fixed $d/2$-dimensional subspace of ${\mathbf{C}}^d$, let $U_1,\ldots,U_d$ be random independent Haar-distributed unitaries on ${\mathbf{C}}^d$, and
define the random states $\rho_i=U_i\frac{P_E}{d/2}U_i^\dagger$ and $\sigma_i=U_i\frac{P_{E^{\perp}}}{d/2}U_i^\dagger$, $1\leq i\leq d$, on ${\mathbf{C}}^d$ (where $P_E$ and $P_{E^{\perp}}$ denote the orthogonal projections onto $E$ and $E^{\perp}$ respectively).
Then, denoting by $\{\ket{1},\ldots,\ket{d}\}$ an orthonormal basis of ${\mathbf{C}}^d$, define
\[ \rho=\frac{1}{d}\sum_{i=1}^d\ket{i}\bra{i}\otimes\rho_i\ \ \textnormal{and}\ \ \sigma=\frac{1}{d}\sum_{i=1}^d\ket{i}\bra{i}\otimes\sigma_i. \]
The pair $(\rho,\sigma)$ satisfies \ref{equation:small-LO} with high probability.

Theorem \ref{theorem:LO-vs-LOCC} is proved in Section 4. It is built on the idea that, typically, a single POVM cannot succeed
simultaneously in several ``sufficiently different'' discrimination tasks.

\section{On the complexity of the class of all POVMs}
\label{sec:all-POVMs}

In this section, we determine how many distinct POVMs a set ${\mathbf{\underline{M}}}$ of POVMs on ${\mathbf{C}}^d$ must contain in order
to approximate the set ${\mathbf{\underline{ALL}}}$ of all POVMs on ${\mathbf{C}}^d$ (in the sense that
$\lambda \|\cdot\|_{{\mathbf{\underline{ALL}}}}\leq\|\cdot\|_{{\mathbf{\underline{M}}}}\leq\|\cdot\|_{{\mathbf{\underline{ALL}}}}$ for some $0 < \lambda < 1$).

The reason for the $\exp(d^2)$ scaling in the first part of Theorem \ref{theorem:approximation-of-ALL} is that these POVMs should be able to discriminate any two states within the family of
states
$\{\frac{1}{\dim E} P_E\}$, where $E$ varies among all subspaces of ${\mathbf{C}}^d$, and $P_E$ denotes the orthogonal projection onto $E$.
The set of $k$-dimensional subspaces of ${\mathbf{C}}^d$
has dimension $k(d-k)$, which is of order $d^2$ when $k$ is proportional to $d$.

The second part of Theorem \ref{theorem:approximation-of-ALL} requires an extra ingredient, since
a single POVM may be able to discriminate exponentially many pairs of subspaces. The concept of mean width (see Appendix \ref{ap:convex-geometry}) provides a neat answer to this problem.

\medskip

To begin with, we prove the first part of Theorem \ref{theorem:approximation-of-ALL}.
Note that the condition $\|\cdot\|_{{\mathbf{\underline{M}}}} \geq (1-\e) \|\cdot\|_{{\mathbf{\underline{ALL}}}}$
is equivalent to $K_{{\mathbf{\underline{M}}}} \supset (1-\e) [-{\mathrm{Id}}, {\mathrm{Id}}]$, the set $K_{{\mathrm{M}}}$ being defined in \ref{eq:definition-K}.
We thus only have to make use of the well-known lemma below.

\begin{lemma}[Approximation of convex bodies by polytopes] \label{lemma:approximation-polytope}
Given a symmetric convex body $K \subset {\mathbf{R}}^n$ and $0<\e<1$, there is a finite family $(x_i)_{i \in I}$ such that $\card (I) \leq (3/\e)^n$ and
\[ (1-\e) K \subset \conv \{ \pm x_i \ : \ i \in I \} \subset K .\]
\end{lemma}

\begin{proof}{
Let ${\mathcal{N}}$ be $\e$-net in $K$, with respect to $\|\cdot\|_K$ (the gauge of $K$, as defined in Appendix \ref{ap:convex-geometry}).
A standard volumetric argument (see e.g. \cite{Pisier}, Lemma 4.10) shows that we may ensure that $\card ({\mathcal{N}})
\leq  (3/\e)^n$. Let $P := \conv ( \pm {\mathcal{N}}) \subset K$.
Given any $x \in K$, there exists $x' \in {\mathcal{N}}$ such that $\|x-x'\|_K \leq \e$. Therefore
\[ \|x\|_P \leq \|x'\|_P+\|x-x'\|_P \leq 1 + \e A, \]
where $A := \sup \{ \|y\|_P \ : \ y \in K\}$. Taking supremum over $x \in K$, we obtain $A \leq 1 +\e A$ and therefore ($A$ is easily seen to be finite)
$A \leq (1-\e)^{-1}$. We thus proved the inequality $\|\cdot\|_P \leq (1-\e)^{-1} \|\cdot\|_K$, which is equivalent to the inclusion $(1-\e)K \subset P$.
}\end{proof}

When applied to the $d^2$-dimensional convex body $K_{{\mathbf{\underline{ALL}}}} = [-{\mathrm{Id}},{\mathrm{Id}}]$, Lemma \ref{lemma:approximation-polytope} implies that there is a
finite family
$(A_i)_{i \in I} \subset [-{\mathrm{Id}},{\mathrm{Id}}]$ with $\card (I) \leq (3/\e)^{d^2}$ and $\conv \{ \pm A_i \ : \ i \in I \} \supset (1-\e) [-{\mathrm{Id}},{\mathrm{Id}}]$.
For every $i \in I$, we may consider the POVM
\[ {\mathrm{M}}_i := \left( \frac{{\mathrm{Id}}+A_i}{2} , \frac{{\mathrm{Id}}-A_i}{2} \right) .\]
If we denote ${\mathbf{\underline{M}}} := \{ {\mathrm{M}}_i \ : \ i \in I \}$, then for any $i \in I$, $\pm A_i \in K_{{\mathrm{M}}_i}$ and therefore
$(1-\e) [-{\mathrm{Id}},{\mathrm{Id}}] \subset K_{{\mathbf{\underline{M}}}}$, which is precisely what we wanted to prove.

\medskip

We now show the second part of Theorem \ref{theorem:approximation-of-ALL}. The key observation is the following lemma, where
we denote by $\alpha_n$ the mean width of a segment $[-x,x]$ for $x$ a unit vector in ${\mathbf{R}}^n$, so that $\alpha_n \sim \sqrt{2/\pi n}$ (see Appendix \ref{ap:convex-geometry}).

\begin{lemma} \label{lemma:mean-width-POVM}
Let ${\mathrm{M}}$ be a POVM on ${\mathbf{C}}^d$. Then the mean width of the set $K_{{\mathrm{M}}}$ defined in \ref{eq:definition-K} satisfies
$w(K_{{\mathrm{M}}}) \leq d \alpha_{d^2}$,
with equality if  ${\mathrm{M}}$ is a rank-$1$ POVM (note that $d \alpha_{d^2}$ is of order $1$).
\end{lemma}

It may be pointed out that the assertion of Lemma \ref{lemma:mean-width-POVM} implies that, as far as the mean width is concerned, all rank-$1$
POVMs are comparable!

\begin{proof}{
Given any POVM ${\mathrm{M}}$, there is a rank-$1$ POVM ${\mathrm{M}}'$ such that $K_{{\mathrm{M}}} \subset K_{\mathrm{M'}}$
(this is easily seen by splitting the POVM
elements from ${\mathrm{M}}$ as a sum of rank-$1$ operators). Therefore, it suffices to show that $w(K_{{\mathrm{M}}})=d \alpha_{d^2}$ for any rank-$1$ POVM.
Let ${\mathrm{M}}=\left(p_i\ketbra{\psi_i}{\psi_i}\right)_{i\in I}$ be a rank-$1$ POVM, where $(p_i)_{i\in I}$ are positive numbers and
$(\psi_i)_{i\in I}$ are unit vectors such that
\[ \sum_{i\in I} p_i\ketbra{\psi_i}{\psi_i} ={\mathrm{Id}} . \]
By taking the trace, we check that the total mass of $\{ p_i \ : \ i\in I\}$ equals $d$. We then have, for any $\Delta \in \cH({\mathbf{C}}^d)$,
\[ h_{K_{{\mathrm{M}}}}(\Delta) = \sum_{i\in I} p_i |\langle \psi_i |\Delta| \psi_i \rangle|. \]
Hence, denoting by $S_{\cH({\mathbf{C}}^d)}$ the Hilbert--Schmidt unit sphere of $\cH({\mathbf{C}}^d)$ (which has dimension $d^2-1$) equipped with the uniform
measure $\sigma$, the mean width of
$K_{{\mathrm{M}}}$ can be computed as
\[ w(K_{{\mathrm{M}}}) = \int_{S_{\cH({\mathbf{C}}^d)}} h_{K_{{\mathrm{M}}}}(\Delta) \, {\mathrm{d}} \sigma(\Delta) =
\sum_{i\in I} p_i \left( \int_{S_{\cH({\mathbf{C}}^d)}} |\langle \psi_i |\Delta| \psi_i \rangle|  \, {\mathrm{d}} \sigma(\Delta)  \right)
 = \sum_{i\in I} p_i \alpha_{d^2} = d \alpha_{d^2} . \]
}\end{proof}

%We are therefore in position to apply Lemma \ref{lemma:mean-width-union-bound}.
Assume that ${\mathbf{\underline{M}}}$ is a family of $N$ POVMs such that
$\|\Delta\|_{{\mathbf{\underline{M}}}} \geq \e \|\Delta\|_1$ for any $\Delta \in \cH({\mathbf{C}}^d)$. This implies that $K_{{\mathbf{\underline{M}}}} \supset \e [-{\mathrm{Id}},{\mathrm{Id}}]$ and therefore
that
\begin{equation} \label{wKM-lower-bound} w(K_{{\mathbf{\underline{M}}}}) \geq \e w([-{\mathrm{Id}},{\mathrm{Id}}]) \simeq \e \sqrt{d},
\end{equation}
where we used last the estimate on the mean width of $[-{\mathrm{Id}},{\mathrm{Id}}]$ (from Theorem \ref{theorem:operator-norm}).
On the other hand,  we have
\begin{equation} \label{eq:convex-hull} K_{{\mathbf{\underline{M}}}} = \conv \left( \bigcup_{{\mathrm{M}} \in {\mathbf{\underline{M}}}} K_{{\mathrm{M}}} \right) ,\end{equation}
so that $K_{{\mathbf{\underline{M}}}}$ is
the convex hull of $N$ sets, each of them of mean width bounded by an absolute constant (by Lemma \ref{lemma:mean-width-POVM}).
We may apply Lemma \ref{lemma:mean-width-union-bound} with $\lambda = \sqrt{d}$
since $[-{\mathrm{Id}},{\mathrm{Id}}]$ is contained in the Hilbert--Schmidt ball of radius $\sqrt{d}$. Recalling that the ambient dimension is $n=d^2$, we get
\begin{equation} \label{wKM-upper-bound} w(K_{{\mathbf{\underline{M}}}}) \leq C \left( 1 + \frac{\sqrt{\log N}}{\sqrt{d}} \right) .\end{equation}
A comparison of the bounds \ref{wKM-lower-bound} and \ref{wKM-upper-bound} immediately yields $\log N \,\succeq\, \e^2 d^2$, as required.

\section{Unbounded gap between ${\mathbf{\underline{LO}}}$ and ${\mathbf{\underline{LOCC}}}$}
\label{section:LO-vs-LOCC}

In this section we give a proof of Theorem \ref{theorem:LO-vs-LOCC}.
Let $\{\ket{1},\ldots,\ket{d}\}$ be an orthonormal basis of ${\mathbf{C}}^d$. For $d$ even, we consider a fixed $d/2$-dimensional subspace $E \subset {\mathbf{C}}^d$,
and denote $\Delta_0 = 2P_E - {\mathrm{Id}}$. We then pick $U_1,\ldots,U_d$ random independent Haar-distributed unitaries on ${\mathbf{C}}^d$, and for $1 \leq i \leq d$ we
consider the random operators $\Delta_i=U_i\Delta_0U_i^\dagger$. We finally introduce
\begin{equation} \label{eq:def-Delta} \Delta = \sum_{i=1}^d \ket{i}\bra{i}\otimes\Delta_i. \end{equation}

For each $1\leq i\leq d$, let ${\mathrm{M}}_i=(M_i,{\mathrm{Id}}-M_i)$ be a POVM on ${\mathbf{C}}^d$ such that $\|\Delta_i\|_{{\mathrm{M}}_i} = \|\Delta_i\|_1$. Then,
\[ {\mathrm{M}} = \left( \ketbra{i}{i} \otimes M_i, \ketbra{i}{i} \otimes ({\mathrm{Id}}-M_i) \right)_{1 \leq i \leq d} \]
is a POVM on ${\mathbf{C}}^d\otimes{\mathbf{C}}^d$ which is in ${\mathbf{\underline{LOCC}^\rightarrow}}$, and therefore
\[ \|\Delta\|_{{\mathrm{M}}} = \|\Delta\|_{{\mathbf{\underline{LOCC}^\rightarrow}}} =  \|\Delta\|_1 = \sum_{i=1}^d  \|\Delta_i\|_1 = d^2 . \]
Theorem \ref{theorem:LO-vs-LOCC} will follow (with $\rho$ and $\sigma$ being the positive and negative parts of $\Delta$,
after renormalization) if we prove that $\|\Delta\|_{{\mathbf{\underline{LO}}}} \leq Cd^{3/2}$ with high probability.

\begin{proposition} \label{prop:discretization-LO}
For $\Delta \in \cH({\mathbf{C}}^d \otimes {\mathbf{C}}^d)$ defined as in \ref{eq:def-Delta}, we have
\begin{equation} \label{eq:LO-Delta} \|\Delta\|_{{\mathbf{\underline{LO}}}} =
\sup\left\{ \sum_{i=1}^d  \|\Delta_i\|_{{\mathrm{N}}} \ : \ {\mathrm{N}}\ \ \textnormal{POVM on}\ \ {\mathbf{C}}^d \right\}. \end{equation}
This quantity can be upper bounded as follows, where ${\mathcal{N}}$ denotes a $\frac{1}{16}$-net in $S_{{\mathbf{C}}^d}$,
\begin{eqnarray}
\label{eq:bound-LO}
\|\Delta\|_{{\mathbf{\underline{LO}}}} & \leq & d \sup_{x \in S_{{\mathbf{C}}^d}} \sum_{i=1}^d \left| \bra{x} \Delta_i \ket{x} \right| \\
\label{eq:bound-LO-2}
& \leq & 2d \sup_{x \in {\mathcal{N}}} \sum_{i=1}^d \left| \bra{x} \Delta_i \ket{x} \right|.
\end{eqnarray}
\end{proposition}

\begin{proof}{
{ The inequality $\geq$ in \ref{eq:LO-Delta} follows by considering the ${\mathbf{\underline{LO}}}$ POVM $(\ketbra{i}{i})_{1 \leq i \leq d} \otimes {\mathrm{N}}$.
Conversely, given
POVMs ${\mathrm{M}}=(M_j)_{j\in J}$ and ${\mathrm{N}}=(N_k)_{k\in K}$ on ${\mathbf{C}}^d$, we have
\begin{eqnarray*}
\|\Delta\|_{{\mathrm{M}}\otimes{\mathrm{N}}} & = & \sum_{j\in J,\ k\in K} \left|\sum_{i=1}^d {\mathrm{Tr}}\left(\left(\ket{i}\bra{i}\otimes\Delta_i\right)
\left(M_j\otimes N_k\right)\right)\right|\\
& \leq & \sum_{i=1}^d \left(\sum_{j\in J}\left|\bra{i}M_j\ket{i}\right|\right)
\left(\sum_{k\in K}\left|{\mathrm{Tr}}\left(\Delta_i N_k\right)\right|\right)\\
& \leq & \sum_{i=1}^d  \|\Delta_i\|_{{\mathrm{N}}},
\end{eqnarray*}
the last inequality being because, for each $1\leq i\leq d$, $\sum_{j\in J}\left|\bra{i}M_j\ket{i}\right| =\sum_{j\in J}\bra{i}M_j\ket{i}
=\braket{i}{i} =1$. Taking the supremum over ${\mathrm{M}}$ and ${\mathrm{N}}$ gives the inequality $\leq$ in  \ref{eq:LO-Delta}.

The supremum in \ref{eq:LO-Delta} is unchanged when restricting to the supremum on POVMs whose elements have rank $1$, since splitting the POVM elements
as sum of rank $1$ operators does not decrease the distinguishability norm. If ${\mathrm{N}}$ is such a POVM, its elements can be written as
$(\alpha_k \ketbra{x_k}{x_k})_{k \in K}$, where $(x_k)_{k \in K}$ are unit vectors and $(\alpha_k)_{k \in K}$ positive numbers satisfying $\sum_{k \in K} \alpha_k=d$.
We thus have in that case
\[ \sum_{i=1}^d \|\Delta_i\|_{{\mathrm{N}}} = \sum_{i=1}^d \sum_{k \in K} |{\mathrm{Tr}}(\Delta_i \cdot \alpha_k \ketbra{x_k}{x_k} )|
\leq d \sup_{x \in S_{{\mathbf{C}}^d}} \sum_{i=1}^d \left| \bra{x} \Delta_i \ket{x} \right|, \]
proving \ref{eq:bound-LO}.

To prove \ref{eq:bound-LO-2}, we introduce the function $g$ defined for $x,y \in {\mathbf{C}}^d$ by $g(x,y) = \sum_{i=1}^d \left|\bra{x}\Delta_i \ket{y}\right|$, and the function $f$ defined for $x \in {\mathbf{C}}^d$ by $f(x)=g(x,x)$. Denote by $G$ the supremum of $g$ over $S_{{\mathbf{C}}^d} \times S_{{\mathbf{C}}^d}$, by
$F$ the supremum of $f$ over $S_{{\mathbf{C}}^d}$ and by $F'$ the supremum of $f$ over a $\delta$-net ${\mathcal{N}}$.
For any $x,y\in{\mathbf{C}}^d$ and $\Delta\in{\mathcal{H}}({\mathbf{C}}^d)$, we have by the polarisation identity
\[\bra{x}\Delta\ket{y} = \frac{1}{4}\left(\bra{x+y}\Delta\ket{x+y}
+i\bra{x+iy}\Delta\ket{x+iy} -\bra{x-y}\Delta\ket{x-y} -i\bra{x-iy}\Delta\ket{x-iy}\right), \]
so that $g(x,y) \leq \frac{1}{4}\left(f(x+y)+f(x+iy)+f(x-y)+f(x-iy)\right)$ and therefore $G\leq 4F$.

Given $x\in S_{{\mathbf{C}}^d}$, there exists $x'\in{\mathcal{N}}$ such that $\|x-x'\|_2 \leq\delta$, and by the triangle inequality, for any $\Delta\in{\mathcal{H}}({\mathbf{C}}^d)$,
\[ \left|\bra{x}\Delta\ket{x}\right| \leq \left|\bra{x}\Delta\ket{x-x'}\right| + \left|\bra{x-x'}\Delta\ket{x'}\right| +
\left|\bra{x'}\Delta\ket{x'}\right| .\]
Summing over $i$ with $\Delta=\Delta_i$ and taking supremum over $x \in S_{{\mathbf{C}}^d}$ gives
\[ F \leq 2 \delta G + F' \leq 8 \delta F +F' .\]
For $\delta=1/16$, we obtain $F \leq 2F'$, and therefore \ref{eq:bound-LO-2} follows from \ref{eq:bound-LO}.}
}\end{proof}

To bound $\|\Delta\|_{{\mathbf{\underline{LO}}}}$, we combine Proposition \ref{prop:discretization-LO} with the following result.

\begin{proposition}
\label{prop:individual}
Let $x$ be a fixed unit vector in ${\mathbf{C}}^d$, $E$ be a fixed $d/2$-dimensional subspace of ${\mathbf{C}}^d$ and $\Delta_0 = 2P_E -{\mathrm{Id}}$, $(U_i)_{1 \leq i \leq n}$ be Haar-distributed independent random unitaries on ${\mathbf{C}}^d$, and for each $1\leq i\leq n$, set
$\Delta_i=U_i\Delta_0U_i^\dagger$. Then, for any $t>1$,
\[ {\mathbf{P}}\left(\sum_{i=1}^n \left|\bra{x}\Delta_i\ket{x}\right|\geq (1+t) n \E |\bra{x} \Delta_1 \ket{x}|\right)\leq e^{-c_0 nt},\]
$c_0$ being a universal constant.
\end{proposition}

\begin{proof}{
Proposition \ref{prop:individual} is a consequence of Proposition 6.2 from \cite{AL}
(which is itself a variation on Bernstein inequalities).
The quantity $\E |\bra{x} \Delta_1 \ket{x}|$ is equal to the so-called ``uniform norm'' of $\Delta_1$
(see \cite{MWW,AL}) and we use the bound from \cite{LW}
\[ \E |\bra{x} \Delta_1 \ket{x}| \leq \frac{1}{d} \|\Delta_1\|_{2} = \frac{1}{\sqrt{d}} . \]
}\end{proof}

We now complete the proof of Theorem \ref{theorem:LO-vs-LOCC}.
Let ${\mathcal{N}}$ be a minimal $1/16$-net in $S_{{\mathbf{C}}^d}$, so that $\card({\mathcal{N}})\leq 48^{2d}$ (see \cite{Pisier}, Lemma 4.10).
Using Propositions \ref{prop:discretization-LO} and \ref{prop:individual} (for $n=d$), and the union bound, we obtain that for any $t>1$
\[ {\mathbf{P}} \left( \|\Delta\|_{{\mathbf{\underline{LO}}}} \geq 2(1+t) d^{3/2} \right) \leq
{\mathbf{P}}\left( \exists \ x\in {\mathcal{N}} \ : \ \sum_{i=1}^d |\bra{x} \Delta_i \ket{x}| \geq (1+t)\sqrt{d}\right) \leq
48^{2d} e^{-c_0dt} . \]

This estimate is less than $1$ when $t$ is larger than some number $t_0$. This shows that
$\|\Delta\|_{{\mathbf{\underline{LO}}}} \leq 2(1+t_0)d^{3/2}$ with high probability while $\|\Delta\|_{{\mathbf{\underline{LOCC}}}^{\rightarrow}} = d^2$, and Theorem \ref{theorem:LO-vs-LOCC} follows.

\begin{remark}
The operator $\Delta$ defined by equation \ref{eq:def-Delta} can we rewritten as $\Delta = d^2(\rho' -
\mathrm{Id}/d^2)$, with
\[ \rho' = \frac{2}{d^2} \sum_{i=1}^d \ketbra{i}{i}
\otimes U_iP_{E}U_i^{\dagger}. \]
It thus follows from Theorem \ref{theorem:LO-vs-LOCC} that $\| \rho' - {\mathrm{Id}}/d^2
\|_{\mathbf{\underline{LO}}}\leq C/\sqrt{d}$ with high probability, while
$\| \rho' - {\mathrm{Id}}/d^2\|_{\mathbf{\underline{LOCC}}^{\rightarrow}}=1$. This property is
characteristic of data locking states. These are states whose accessible
mutual information (i.e. the maximum classical mutual information that can
be achieved by local measurements) drastically underestimates their
quantum mutual information (see \cite{DVHLST} for the original description
of this phenomenon). Now, following \cite{DFHL} and \cite{FHS}, data locking
may also be defined in terms of distinguishability from the maximally
mixed state by local measurements: informally, a state $\rho$ on
${\mathbf{C}}^d\otimes{\mathbf{C}}^d$ which is such that $\| \rho -
{\mathrm{Id}}/d^2\|_{\mathbf{\underline{LO}}} \ll \|\rho-
{\mathrm{Id}}/d^2\|_{\mathbf{\underline{LOCC}}^{\rightarrow}}$ may be used
for information locking.
\end{remark}

\section{Generic unbounded gap between ${\mathbf{\underline{SEP}}}$ and ${\mathbf{\underline{PPT}}}$}
\label{sec:local-POVM-proofs}

\subsection{Volume and mean width estimates}

The first step towards Theorem \ref{theorem:typical-states} is to estimate globally the size of the (dual) unit balls $K_{{\mathbf{\underline{PPT}}}}$, $K_{{\mathbf{\underline{SEP}}}}$ and $K_{{\mathbf{\underline{LOCC}^\rightarrow}}}$ associated to the
measurement norms $\|\cdot\|_{{\mathbf{\underline{PPT}}}}$, $\|\cdot\|_{{\mathbf{\underline{SEP}}}}$ and $\|\cdot\|_{{\mathbf{\underline{LOCC}^\rightarrow}}}$. Classical useful invariants used to quantify the size of convex
bodies include the volume radius and
the mean width, which are defined in Appendix \ref{ap:convex-geometry}.

Note that whenever we use tools from convex geometry in the space ${\mathcal{H}}({\mathbf{C}}^d \otimes {\mathbf{C}}^d)$ (which has dimension $d^4$) it is tacitly
understood that we use the Euclidean structure
induced by the
Hilbert--Schmidt inner product $\langle A,B \rangle = {\mathrm{Tr}}(AB)$.
The definitions of the volume radius and the mean width of $K_{{\mathbf{\underline{M}}}}$ thus become
\[ \vrad (K_{{\mathbf{\underline{M}}}}) = \left( \frac{\vol K_{{\mathbf{\underline{M}}}}}{\vol B_{HS}} \right)^{1/d^4} \]
and
\[ w (K_{{\mathbf{\underline{M}}}}) = \int_{S_{HS}} \| \Delta \|_{{\mathbf{\underline{M}}}} \, {\mathrm{d}} \sigma (\Delta) , \]
where $B_{HS}$ denotes the Hilbert--Schmidt unit ball of ${\mathcal{H}}({\mathbf{C}}^d \otimes {\mathbf{C}}^d)$ and $S_{HS}$ its Hilbert--Schmidt unit sphere equipped
with the uniform measure $\sigma$.
Here are the estimates on the volume radius and the mean width of $K_{{\mathbf{\underline{PPT}}}}$, $K_{{\mathbf{\underline{SEP}}}}$ and $K_{{\mathbf{\underline{LOCC}^\rightarrow}}}$. As a reference, recall that
(on ${\mathbf{C}}^d \otimes {\mathbf{C}}^d$)
\[ \vrad(K_{{\mathbf{\underline{ALL}}}}) \simeq w(K_{{\mathbf{\underline{ALL}}}}) \simeq d .\]
This follows from Theorem \ref{theorem:operator-norm} once we have in mind that $K_{{\mathbf{\underline{ALL}}}} = [-{\mathrm{Id}},{\mathrm{Id}}]$.

\begin{theorem}
\label{theorem:vrad-w-PPT-SEP}
In ${\mathbf{C}}^d \otimes {\mathbf{C}}^d$, one has
\[ {\mathrm{vrad}}\left(K_{{\mathbf{\underline{PPT}}}}\right)\simeq w\left(K_{{\mathbf{\underline{PPT}}}}\right)\simeq d,\]
and
\[ {\mathrm{vrad}}\left(K_{{\mathbf{\underline{LOCC}^\rightarrow}}}\right) \simeq w \left(K_{{\mathbf{\underline{LOCC}^\rightarrow}}}\right) \simeq \sqrt{d},\]
\[ {\mathrm{vrad}}\left(K_{{\mathbf{\underline{LOCC}}}}\right) \simeq w \left(K_{{\mathbf{\underline{LOCC}}}}\right)\simeq \sqrt{d},\]
\[ {\mathrm{vrad}}\left(K_{{\mathbf{\underline{SEP}}}}\right) \simeq w \left(K_{{\mathbf{\underline{SEP}}}}\right)\simeq \sqrt{d}. \]
\end{theorem}

To prove these results, we will make essential use of the Urysohn inequality (Theorem \ref{theorem:urysohn}): for any convex body
$K \subset {\mathbf{R}}^n$, we have ${\mathrm{vrad}}(K) \leq w(K)$. In particular, Theorem \ref{theorem:vrad-w-PPT-SEP} follows from the following four
inequalities: (a) $w(K_{{\mathbf{\underline{PPT}}}})\,\preceq\, d$, (b) $\vrad(K_{{\mathbf{\underline{PPT}}}}) \,\succeq\, d$, (c) $w(K_{{\mathbf{\underline{SEP}}}})\,\preceq\, \sqrt{d}$
(d) $\vrad(K_{{\mathbf{\underline{LOCC}^\rightarrow}}}) \,\succeq\, \sqrt{d}$.

\subsection{(a) Proof that $w(K_{{\mathbf{\underline{PPT}}}})\,\preceq\, d$} This follows from the inclusion $K_{{\mathbf{\underline{PPT}}}} \subset [-{\mathrm{Id}},{\mathrm{Id}}]$, together with the estimate on the mean width of $[-{\mathrm{Id}},{\mathrm{Id}}]$
 from Theorem \ref{theorem:operator-norm}.

\subsection{(b) Proof that $\vrad(K_{{\mathbf{\underline{PPT}}}})\,\succeq\, d$} We start by noticing that
\[K_{{\mathbf{\underline{PPT}}}}=[-{\mathrm{Id}},{\mathrm{Id}}]\cap [-{\mathrm{Id}},{\mathrm{Id}}]^{\Gamma}.\]
We apply the Milman--Pajor inequality (Corollary \ref{corollary:Milman-Pajor}) to the convex body $[-{\mathrm{Id}},{\mathrm{Id}}]$
(which indeed has the origin as center of mass) and to the orthogonal transformation $\Gamma$ (the partial transposition). This yields
\[{\mathrm{vrad}}\left(K_{{\mathbf{\underline{PPT}}}}\right)\geq \frac{1}{2}\frac{{\mathrm{vrad}}\left([-{\mathrm{Id}},{\mathrm{Id}}]\right)^2}{w\left([-{\mathrm{Id}},{\mathrm{Id}}]\right)}\simeq d,\]
where we used the estimates on the volume radius and the mean width of $[-{\mathrm{Id}},{\mathrm{Id}}]$ from Theorem \ref{theorem:operator-norm}.

\subsection{(c) Proof that $w(K_{{\mathbf{\underline{SEP}}}})\,\preceq\, \sqrt{d}$} We are going to relate $K_{{\mathbf{\underline{SEP}}}}$ with the set ${\mathcal{S}}$ of
separable states on ${\mathbf{C}}^d \otimes {\mathbf{C}}^d$. In fact, denoting the cone with base ${\mathcal{S}}$ by
\[ {\mathbf{R}}^+{\mathcal{S}} := \{ \lambda \rho \ : \ \lambda \in {\mathbf{R}}^+,\ \rho \in {\mathcal{S}} \}, \]
we have $K_{{\mathbf{\underline{SEP}}}} = L \cap (-L)$, where
\[ L := 2\left({\mathbf{R}}^+{\mathcal{S}}\cap[0,{\mathrm{Id}}]\right)-{\mathrm{Id}} .\]
This gives immediately an upper bound on the mean width of $K_{{\mathbf{\underline{SEP}}}}$
\[ w(K_{{\mathbf{\underline{SEP}}}}) \leq w(L) \leq 2 w({\mathbf{R}}^+{\mathcal{S}}\cap[0,{\mathrm{Id}}]) \leq 2 w(\{\lambda\rho \ : \ \lambda\in[0,d^2],\ \rho\in{\mathcal{S}}\})
=2d^2 w(\conv(\{0\},{\mathcal{S}})).\]
Now, if $K,K'$ are two convex sets such that $K \cap K' \neq \emptyset$, then $w(\conv(K,K')) \leq w(K) + w(K')$. So, denoting  by
$\alpha_n$ the mean width of a segment $[-x,x]$ for $x$ a unit vector in ${\mathbf{R}}^n$, we have
\[ w(\conv(\{0\},{\mathcal{S}})) \leq w(\conv \{0,{\mathrm{Id}}/d^2 \}) + w ({\mathcal{S}}) \,\preceq\, \frac{\alpha_{d^4}}{d} + \frac{1}{d^{3/2}} \,\preceq\,
\frac{1}{d^{3/2}}, \]
where we used the estimate $w({\mathcal{S}}) \simeq d^{-3/2}$ from Theorem \ref{theorem:separable-states}, and the fact that $\alpha_n\simeq n^{-1/2}$
(see Appendix \ref{ap:convex-geometry}).

\subsection{(d) Proof that $\vrad(K_{{\mathbf{\underline{LOCC}^\rightarrow}}})\,\succeq\, \sqrt{d}$} We consider the following set of states on ${\mathbf{C}}^d \otimes {\mathbf{C}}^d$
\[ T = \conv \left\{ \ketbra{\psi}{\psi} \otimes \sigma \ : \ \psi \in S_{{\mathbf{C}}^d}, \ \sigma
\textnormal{ a state on }{\mathbf{C}}^d \textnormal{ such that } \|\sigma\|_{\infty} \leq 3/d \right\}. \]
A connection between $T$ and ${\mathbf{\underline{LOCC}^\rightarrow}}$ is given by the following lemma.

\begin{lemma} \label{lemma:LOCC}
Let $\rho,\rho' \in T$ such that $\rho+\rho'=2 {\mathrm{Id}}/d^2$. Then the operators
$\frac{d^2}{6} \rho $ and $\frac{d^2}{6} \rho'$ belong to $K_{{\mathbf{\underline{LOCC}^\rightarrow}}}$.
\end{lemma}

\begin{proof}{
There exist convex combinations $(\alpha_i)_{i \in I}, (\alpha'_j)_{j \in J}$, unit vectors $(\psi_i)_{i \in I}, (\psi'_j)_{j \in J}$
and states $(\sigma_i)_{i \in I},(\sigma'_j)_{j \in J}$ satisfying $\|\sigma_i\|_{\infty} \leq 3/d$, $\|\sigma'_j\|_{\infty} \leq 3/d$, such that
\[ \rho = \sum_{i \in I} \alpha_i \ketbra{\psi_i}{\psi_i} \otimes \sigma_i \ \ \textnormal{ and } \ \ \rho' = \sum_{j \in J} \alpha'_j
\ketbra{\psi'_j}{\psi'_j} \otimes \sigma'_j .\]
Define states $(\tau_i)_{i \in I}$ and $(\tau'_j)_{j \in J}$ by the relations $\sigma_i + 2\tau_i = \sigma'_j + 2\tau'_j= 3{\mathrm{Id}}/d$. It can then be checked that the following POVM is in
${\mathbf{\underline{LOCC}^\rightarrow}}$
\[ {\mathrm{M}} = \left( \frac{d^2}{6} \alpha_i \ketbra{\psi_i}{\psi_i} \otimes \sigma_i,
\frac{d^2}{6} \alpha_i \ketbra{\psi_i}{\psi_i} \otimes 2\tau_i, \frac{d^2}{6} \alpha'_j \ketbra{\psi'_j}{\psi'_j} \otimes \sigma'_j,
\frac{d^2}{6} \alpha'_j \ketbra{\psi'_j}{\psi'_j} \otimes 2\tau'_j \right)_{i\in I,j \in J}. \]
Hence, the operators $\frac{d^2}{6} \rho$ and $\frac{d^2}{6} \rho'$ belong to $K_{{\mathrm{M}}}$ and therefore to $K_{{\mathbf{\underline{LOCC}^\rightarrow}}}$.
}\end{proof}

Let $\widetilde{T}$ be the symmetrization of $T$ defined as $\widetilde{T} = T \cap \left\{ 2 \frac{{\mathrm{Id}}}{d^2} - T \right\}$.
By Lemma \ref{lemma:LOCC} and the fact that $K_{{\mathbf{\underline{LOCC}}}^\rightarrow}$ is centrally symmetric, we have
\[ \frac{d^2}{6} \conv ( \widetilde{T}, -\widetilde{T} ) \subset K_{{\mathbf{\underline{LOCC}}}^\rightarrow} .\]

We are going to give a lower
bound on the volume radius of $\widetilde{T}$.
The center of mass of the set $T$ equals the maximally mixed state ${\mathrm{Id}}/d^2$ (indeed, the center of mass commutes with local unitaries). By
Corollary \ref{corollary:Milman-Pajor}, this implies that $\vrad( \widetilde{T} ) \geq \frac{1}{2} \vrad(T)$. On the other hand, one has (see definitions in Appendix \ref{ap:standard})
\begin{equation} \label{eq:convT} \conv(T, - T) \supset \frac{1}{d} \cdot S_1^d \hat{\otimes} S_{\infty}^d .\end{equation}
Let us check \ref{eq:convT}. An extreme point of $\frac{1}{d} \cdot S_1^d \hat{\otimes} S_{\infty}^d$ has the form $\pm \ketbra{\psi}{\psi} \otimes A$
for $\psi \in S_{{\mathbf{C}}^d}$ and $A \in \cH({\mathbf{C}}^d)$ such that $\|A\|_{\infty} \leq 1/d$. Let $\e = 2 - \|A\|_{1} \geq 1$ and let $A^+,A^-$ be the positive and
negative parts of $A$. Set $\lambda^{\pm} = \e/4 + {\mathrm{Tr}} A^\pm/2$ (so that $\lambda^++\lambda^-=1$), and consider the states
$\rho^\pm =  \frac{1}{\lambda^\pm} ( \e/4 \cdot {\mathrm{Id}}/d + A^\pm/2)$. We have
\[ \| \rho^{\pm} \|_{\infty} \leq \frac{{\e}/{4d}+ {1}/{2d}}{{\e}/{4}} \leq \frac{3}{d}\]
and therefore $\rho^\pm \in T$. Since $A = \lambda^+ \rho^+ - \lambda^- \rho^-$, this shows \ref{eq:convT}. Using Theorem
\ref{theorem:volume-S1-Sinfini}, it follows that
\[ \vrad( \conv(T,-T) ) \,\succeq\, d^{-3/2} .\]
And therefore,
\[ \vrad( \conv ( \widetilde{T}, -\widetilde{T} ) ) \,\succeq\, \vrad ( \widetilde{T}) \,\succeq\, \vrad(T) \,\succeq\, \vrad( \conv(T,-T) ) \,\succeq\, d^{-3/2} , \]
the first and third inequalities being due to the Rogers--Shephard inequality (Theorem \ref{theorem:rogers-shephard}).
We eventually get
\[ \vrad( K_{{\mathbf{\underline{LOCC}^\rightarrow}}} ) \,\succeq\, \sqrt{d} .\]

\subsection{Discriminating between two generic states}

Let ${\mathbf{\underline{M}}}$ be a family of POVMs on ${\mathbf{C}}^d$ (possibly reduced to a single POVM). We relate the mean width $w(K_{{\mathbf{\underline{M}}}})$ to the typical
performance of ${\mathbf{\underline{M}}}$ for
discriminating two random states, chosen independently and uniformly from the set ${\mathcal{D}}({\mathbf{C}}^d)$ of all states on ${\mathbf{C}}^d$.

\begin{proposition} \label{proposition:Delta-vs-difference-of-states}
Let ${\mathbf{\underline{M}}}$ be a family of POVMs on ${\mathbf{C}}^d$, and denote $\omega := w(P_{H_0} K_{{\mathbf{\underline{M}}}})$, where $P_{H_0}$ stands for the orthogonal
projection onto
the hyperplane $H_0 \subset \cH({\mathbf{C}}^d)$ of trace $0$ Hermitian operators on ${\mathbf{C}}^d$.
Let $\rho$ and $\sigma$ be two random states, chosen independently with respect to the uniform measure on ${\mathcal{D}}({\mathbf{C}}^d)$.
Then,
\begin{equation}
\label{eq:expectation}
\E := \E \| \rho - \sigma \|_{{\mathbf{\underline{M}}}} \simeq \frac{\omega}{\sqrt{d}} .
\end{equation}
Moreover, we have the concentration estimate
\begin{equation}
\label{eq:concentration}
\forall\ t>0,\ {\mathbf{P}} \left( \left| \| \rho - \sigma \|_{{\mathbf{\underline{M}}}} - \E \right| > t \right) \leq 2 \exp (-cdt^2), \end{equation}
$c$ being a universal constant.
\end{proposition}

We first deduce Theorem \ref{theorem:typical-states} from Theorem \ref{theorem:vrad-w-PPT-SEP} and
Proposition \ref{proposition:Delta-vs-difference-of-states} (we warn the reader that we apply the latter on the space ${\mathbf{C}}^d \otimes {\mathbf{C}}^d$,
and therefore the ambient dimension is $d^2$ instead of $d$).

\begin{proof}{[Proof of Theorem \ref{theorem:typical-states}]
Let ${\mathbf{\underline{M}}} \in \{ {\mathbf{\underline{LOCC}}}, {\mathbf{\underline{LOCC}^\rightarrow}},{\mathbf{\underline{SEP}}},{\mathbf{\underline{PPT}}} \}$.
While we computed $w(K_{{\mathbf{\underline{M}}}})$ in Theorem \ref{theorem:vrad-w-PPT-SEP}, the relevant quantity here
is $w(P_{H_0} K_{{\mathbf{\underline{M}}}})$. We show that both are comparable. We first have the upper bound (see \ref{eq:width-projection} from Appendix
\ref{ap:convex-geometry})
\[ w(P_{H_0}K_{{\mathbf{\underline{M}}}}) \,\preceq\, w(K_{{\mathbf{\underline{M}}}}). \]
To get the reverse bound, we consider the volume radius rather than the mean width. If we denote more generally by $H_t$ the hyperplane of
trace $t$ operators on ${\mathbf{C}}^d$, we have by
Fubini's theorem
\[ \vol_{d^4}(K_{{\mathbf{\underline{M}}}}) = \frac{1}{d} \int_{-d^2}^{d^2} \vol_{d^4-1}( K_{{\mathbf{\underline{M}}}} \cap H_t) \,{\mathrm{d}}t .\]
By the Brunn--Minkowski inequality, the function under the integral is maximal when $t=0$, and therefore
\[ \vol_{d^4}(K_{{\mathbf{\underline{M}}}}) \leq 2d \vol_{d^4-1}(K_{{\mathbf{\underline{M}}}} \cap H_0). \]
It follows easily that $w(P_{H_0}K_{{\mathbf{\underline{M}}}})\geq\vrad(P_{H_0} K_{{\mathbf{\underline{M}}}}) \geq \vrad ( K_{{\mathbf{\underline{M}}}} \cap H_0 ) \,\succeq\,
\vrad(K_{{\mathbf{\underline{M}}}}) \simeq w(K_{{\mathbf{\underline{M}}}})$,
the first inequality being the Urysohn inequality (Theorem \ref{theorem:urysohn}) and the last estimate being by Theorem
\ref{theorem:vrad-w-PPT-SEP}.
Once this is known, Theorem \ref{theorem:typical-states}
is immediate from Proposition \ref{proposition:Delta-vs-difference-of-states}.
}\end{proof}

\begin{proof}{[Proof of Proposition \ref{proposition:Delta-vs-difference-of-states}]
We first show the concentration estimate \ref{eq:concentration}, using the following representation due to \.Zyczkowski and
Sommers \cite{ZS}: $\rho$ has the same distribution as
$MM^\dagger$, where $M$ is uniformly distributed on the Hilbert--Schmidt unit sphere (denoted $S_{HS}$) in the space of complex $d \times d$ matrices.
We estimate the Lipschitz
constant of the function $f:(M,N) \mapsto \|MM^\dagger-NN^\dagger\|_{{\mathbf{\underline{M}}}}$, defined on $S_{HS} \times S_{HS}$, as follows:
\begin{eqnarray*}
f(M_1,N_1)-f(M_2,N_2) &=& \| M_1M_1^\dagger-N_1N_1^\dagger\|_{{\mathbf{\underline{M}}}} - \| M_2M_2^\dagger-N_2N_2^\dagger\|_{{\mathbf{\underline{M}}}}\\
&\leq& \| M_1M_1^\dagger-M_2M_2^\dagger\|_{{\mathbf{\underline{M}}}} +
\| N_1N_1^\dagger-N_2N_2^\dagger\|_{{\mathbf{\underline{M}}}} \\
&\leq& \sqrt{d} \left( \| M_1M_1^\dagger-M_2M_2^\dagger\|_{2} + \| N_1N_1^\dagger-N_2N_2^\dagger\|_{2} \right) \\
&\leq& \sqrt{d} \left( 2\|M_1-M_2\|_{2} + 2\|N_1-N_2\|_{2} \right).
\end{eqnarray*}
We used the standard bounds $\|\cdot\|_{{\mathbf{\underline{M}}}}\leq\|\cdot\|_1\leq\sqrt{d}\|\cdot\|_2$ and $||AA^\dagger-BB^\dagger||_{2} \leq
||(A-B)B^\dagger||_{2}+||A(A-B)^\dagger||_{2}$
to get the second and the third inequalities respectively. We obtain as a consequence of Lemma \ref{lemma:levy-two-spheres} below
(a variation on L\'evy's lemma)
the desired estimate
\[ {\mathbf{P}} \left( \left| \|\rho-\sigma\|_{{\mathbf{\underline{M}}}} -\E \right| > t \right) \leq 2 \exp(-cdt^2) .\]
In our application of Lemma \ref{lemma:levy-two-spheres}, we identify the set of complex $d \times d$ matrices with ${\mathbf{R}}^n$ ($n=2d^2$), and use
$L=2\sqrt{d}$.

\begin{lemma}
\label{lemma:levy-two-spheres}
Let $S$ be the unit sphere in ${\mathbf{R}}^n$, and equip $S \times S$ with the metric $d((x,y),(x',y')):=|x-x'|+|y-y'|$ and the measure $\mu \otimes \mu$,
where $\mu$ is the
uniform probability measure on $S$. For any $L$-Lipschitz function $f : S\times S \to {\mathbf{R}}$ and any $t>0$,
\[ {\mathbf{P}}( | f - \E f| > t ) \leq 2 \exp(-cn t^2/L^2), \]
$c$ being a universal constant.
\end{lemma}

Lemma \ref{lemma:levy-two-spheres} can be deduced quickly from the usual L\'evy lemma (see \cite{Levy}) which quantifies the phenomenon of
concentration of measure on the sphere.
If we denote $E_x := \int_{S} f(x,y) \,{\mathrm{d}}\mu(y)$, we may apply L\'evy's lemma
to show that, for fixed $x$, the function $y \mapsto f(x,y)$ concentrates around its expectation $E_x$, and again L\'evy's lemma to show that the function
$x \mapsto E_x$ (which is $L$-Lipschitz,
as an average of $L$-Lipschitz functions) is also well-concentrated.

\medskip

We now prove the first part of Proposition \ref{proposition:Delta-vs-difference-of-states}.
Let $\Delta$ be a random matrix uniformly chosen from the Hilbert--Schmidt sphere in the hyperplane $H_0$, and $\rho,\sigma$ be independent
random states with uniform distribution.
We claim that, from a very rough perspective,
the spectra of $\rho-\sigma$ and $\frac{1}{\sqrt{d}} \Delta$ look similar. More precisely, we have

\begin{lemma} \label{lemma:global-spectra}
Let $\rho,\sigma$ be independent random states uniformly chosen from ${\mathcal{D}}({\mathbf{C}}^d)$, and $\Delta$ be a random
matrix uniformly chosen from the Hilbert--Schmidt sphere in the hyperplane $H_0$. Then with large probability
\[ \|\Delta\|_1 \simeq \sqrt{d},\ \ \|\Delta\|_{2} = 1,\ \ \|\Delta\|_{\infty} \simeq 1/\sqrt{d},\]
\[ \| \rho - \sigma \|_1 \simeq 1,\ \ \|\rho-\sigma\|_{2}\simeq 1/\sqrt{d},\ \ \|\rho-\sigma\|_{\infty} \simeq 1/d.\]
Moreover these statements hold in expectation: e.g. $\E \| \Delta \|_{\infty}\simeq 1/\sqrt{d}$ and $\E \|\rho-\sigma\|_{\infty} \simeq 1/d$.
\end{lemma}

In order to compare $\rho-\sigma$ with $\Delta$, we rely on the following lemma. For $x=(x_1,\dots,x_n) \in {\mathbf{R}}^n$, we denote
$\|x\|_{\infty} = \max \{ |x_i| \ : \ 1 \leq i \leq n \}$ and $\|x\|_1= \sum_{i=1}^n |x_i|$.

\begin{lemma} \label{lemma:comparison-of-norms} Let $E = \{ x \in {\mathbf{R}}^n \ : \ \sum_{i=1}^n x_i =0 \}$ and let
$|||\cdot|||$ be a norm on $E$ which is invariant under permutation of
coordinates.
Then, for any nonzero vectors $x,y \in E$, we have
\begin{equation}
\label{eq:equiv}
 |||x|||  \leq  2n\frac{\|x\|_{\infty}}{\|y\|_1} |||y|||.
\end{equation}
\end{lemma}

Assuming both lemmas, we now complete the proof of Proposition \ref{proposition:Delta-vs-difference-of-states}. On the hyperplane $E \subset {\mathbf{R}}^d$
of vectors whose sum of coordinates is zero,
we define a norm by
\[ |||x||| := \int_{{\mathcal{U}}(d)} \| U \diag(x) U^\dagger \|_{{\mathbf{\underline{M}}}} \, {\mathrm{d}}U , \]
where the integral is taken with respect to the Haar measure on the unitary group, and $\diag(x)$ denotes the diagonal matrix on ${\mathbf{C}}^d$ with diagonal
elements equal to
the coordinates of $x$. Note that $|||\cdot|||$ is obviously invariant under permutation of coordinates. Also, $\Delta$ has the same distribution as
$U \diag\left(\spec(\Delta)\right) U^\dagger$, where $U$ is a Haar-distributed unitary
matrix independent from $\Delta$ and $\spec(A) \in {\mathbf{R}}^d$ denotes the spectrum of $A \in \cH({\mathbf{C}}^d)$ (the ordering of eigenvalues being irrelevant).
The same holds for $\rho-\sigma$ instead of $\Delta$, and it follows that
\[ \E |||\spec(\Delta)||| = \E ||\Delta||_{{\mathbf{\underline{M}}}}  \ \ \ \textnormal{and} \ \ \ \E |||\spec(\rho-\sigma)||| = \E ||\rho - \sigma||_{{\mathbf{\underline{M}}}}. \]

Let us show that
\begin{equation} \label{eq:half} \E ||\rho-\sigma||_{{\mathbf{\underline{M}}}} \simeq \E \frac{1}{\sqrt{d}} ||\Delta||_{{\mathbf{\underline{M}}}}. \end{equation}

We first prove the inequality $\,\preceq\,$. Say that a vector $y \in E$ satisfies the condition $(\star)$ if
$\|y\|_{1} \geq c\sqrt{d}$, where we may choose the constant $c$ such that the random vector $\spec(\Delta)$ satisfies the condition $(\star)$
with probability larger than $1/2$ (this is possible, as we check using Lemma \ref{lemma:global-spectra}).
Now, by Lemma \ref{lemma:comparison-of-norms}, for any $y \in E$ satisfying
condition $(\star)$ and any $x \in E$, we have
\[ |||x||| \,\preceq\, \sqrt{d} \|x\|_{\infty} \cdot |||y||| .\]
We apply this inequality with $x=\spec(\rho-\sigma)$ and take expectation.
This gives (using the statement about expectations in Lemma \ref{lemma:global-spectra})
\[ \E |||\rho-\sigma|||_{{\mathbf{\underline{M}}}} \,\preceq\, \frac{1}{\sqrt{d}} |||y||| .\]
This inequality is true for any $y \in E$ satisfying condition $(\star)$. Therefore,
\begin{eqnarray*}
\E ||\Delta||_{{\mathbf{\underline{M}}}} & = & \E |||\spec(\Delta)|||\\
& \succeq &\, \sqrt{d} \cdot {\mathbf{P}}\big(\spec(\Delta)
\textnormal{ satisfies condition }(\star) \big) \E \|\rho-\sigma\|_{{\mathbf{\underline{M}}}}\\
& \simeq & \sqrt{d} \E \|\rho-\sigma\|_{{\mathbf{\underline{M}}}},
\end{eqnarray*}
as needed. This proves one half of \ref{eq:half}, and the reverse inequality is proved along the exact same lines.
Finally, we note that
\[ \E ||\Delta||_{{\mathbf{\underline{M}}}} = w\left(P_{H_0} K_{{\mathbf{\underline{M}}}} \right),\]
which, together with \ref{eq:half}, shows \ref{eq:expectation}, and concludes the proof.
}\end{proof}

\begin{proof}{[Proof of Lemma \ref{lemma:global-spectra}]
This is folklore in random matrix theory, in fact much more precise results are known
(for example, $\simeq$ can be replaced with $\sim$, with specific constants implicit
in that notation). However, most of the literature focuses on slightly different
random setups. Accordingly, we sketch an essentially self-contained elementary argument for completeness.

First of all, we observe that it is enough to prove the upper estimate for
$\|\cdot\|_{\infty}$ and the lower estimate for $\|\cdot\|_{2}$. Indeed, the
remaining upper estimates and the lower estimate for $\|\cdot\|_{\infty}$
follow then from the generally valid inequalities
$\|\cdot\|_{1} \leq \sqrt{d} \|\cdot\|_{2} \leq d \|\cdot\|_{\infty}$, while
the lower bound for $ \|\cdot\|_{1}$ follows from $\|\cdot\|_{2} \leq \|\cdot\|_{1}^{1/2} \|\cdot\|_{\infty}^{1/2}$.

The upper bound on $\|\cdot\|_{\infty}$ can be proved by a standard net argument. The lower bound on $\|\Delta\|_{2}$ is trivial,
while for $\|\rho -\sigma\|_{2}$ we may proceed as follows. First, using concentration of measure in the form of Lemma \ref{lemma:levy-two-spheres},
$\E \| \rho - \sigma \|_{2}$ is comparable to $\left( \E \| \rho - \sigma \|^2_{2} \right)^{1/2}$. Next, by Jensen inequality,
\[ \E \| \rho - \sigma \|^2_{2} \geq \E \| \rho - {\mathrm{Id}}/d \|^2_{2}. \]
Recalling that $\rho$ can be represented as $MM^\dagger$, with $M$ uniformly distributed on $S_{HS}$, the last
quantity can be expanded as
\[ \E \left\| \rho - \frac{{\mathrm{Id}}}{d} \right\|_{2}^2 = \E {\mathrm{Tr}} |M|^4 - \frac{1}{d}\]
and it can be checked by moments expansion that $\E {\mathrm{Tr}} |M|^4 \sim 2/d$.
}\end{proof}

\begin{proof}{[Proof of Lemma \ref{lemma:comparison-of-norms}]
Define $\alpha = 2n \|x\|_{\infty}/\|y\|_1$. By elementary properties of majorization (see Chapter II in \cite{Bhatia})
it is enough to show that $x$ is majorized by $\alpha y$, i.e. that for every $1 \leq k \leq n$,
\[ \sum_{i=1}^k x_i^\downarrow \leq \alpha \sum_{i=1}^k y_i^\downarrow ,\]
where $(x_i^\downarrow)_{1 \leq i \leq n}, (y_i^\downarrow)_{1 \leq i \leq n}$ denote the non-increasing rearrangement of $x,y$. This follows from
the inequalities
\begin{equation} \label{eq:majorization}
\frac{1}{\|x\|_{\infty}} \sum_{i=1}^k x_i^\downarrow \leq \min (k,n-k) \leq \frac{2n}{\|y\|_1} \sum_{i=1}^k y_i^\downarrow.
\end{equation}
The left-hand inequality in \ref{eq:majorization} follows from the triangle inequality, once we have in mind that
$x_1^\downarrow + \cdots + x_k^\downarrow = -(x_{k+1}^\downarrow + \cdots + x_n^\downarrow)$. To prove the right-hand inequality in \ref{eq:majorization}, note that
the sum of positive coordinates of $y$ and the sum of negative coordinates of $y$ both equal $\|y\|_1/2$. Let $\ell$ be the number of positive coordinates
of $y$. If $k \leq \ell$, then $y_1^\downarrow+\cdots+y_k^\downarrow \geq \frac{k}{\ell} \|y\|_1/2 \geq \frac{k}{2n} \|y\|_1$, while if $k > \ell$,
then $y_1^\downarrow+\cdots+y_k^\downarrow = -(y_{k+1}^\downarrow + \cdots + y_n^\downarrow) \geq \frac{n-k}{n-\ell} \|y\|_1/2 \geq \frac{n-k}{2n} \|y\|_1$.
}\end{proof}

\section{Applications to quantum data hiding}
\label{sec:data-hiding}

\subsection{Bipartite data hiding}

As already mentioned, what Theorem \ref{theorem:typical-states} establishes is that generic bipartite states are data hiding for separable measurements but not for PPT measurements. This fact somehow counterbalances the usually cited constructions of data hiding schemes using Werner states (see e.g. \cite{DVLT1,DVLT2,EW} and \cite{MWW,LW}). Werner states are indeed data hiding in the exact same way for both separable and PPT measurements.
\smallskip

Besides, results in the same vein as those from Theorem \ref{theorem:typical-states} but more specifically orientated towards applications to quantum data hiding may be quite directly written down. In fact, one often thinks of data hiding states as being orthogonal states, hence perfectly distinguishable by the suitable global measurement, that are nevertheless barely distinguishable by any local measurement. The following theorem provides a statement in that direction.

\begin{theorem}
\label{theorem:data-hiding}
There are universal constants $C,c$ such that the following holds. Given a dimension $d$,
let $E$ be a $\frac{d^2}{2}$-dimensional subspace of ${\mathbf{C}}^d\otimes{\mathbf{C}}^d$ (we assume without loss of generality that $d$ is even). Let also
$\rho=\frac{1}{d^2/2}UP_EU^\dagger$ and $\sigma=\frac{1}{d^2/2}UP_{E^{\perp}}U^\dagger$, where $U$ is a Haar-distributed random unitary on ${\mathbf{C}}^d\otimes{\mathbf{C}}^d$. Then,
\[ \| \rho - \sigma \|_{{\mathbf{\underline{ALL}}}}=2, \]
whereas with high probability,
\[ c \leq \| \rho - \sigma \|_{{\mathbf{\underline{PPT}}}} \leq C, \]
\[ \frac{c}{\sqrt{d}} \leq \| \rho - \sigma \|_{{\mathbf{\underline{SEP}}}} \leq \frac{C}{\sqrt{d}}. \]
\end{theorem}

\begin{proof}{
The first part of Theorem \ref{theorem:data-hiding} is clear: the random states $\rho$ and $\sigma$ are orthogonal by construction, so that $\| \rho - \sigma \|_{{\mathbf{\underline{ALL}}}}=\| \rho - \sigma \|_1=2$.

To prove the second part of Theorem \ref{theorem:data-hiding}, the only thing we have to show is that Proposition \ref{proposition:Delta-vs-difference-of-states} also holds for the random states $\rho$ and $\sigma$ considered here.

Now, for any family ${\mathbf{\underline{M}}}$ of POVMs on ${\mathbf{C}}^d\otimes{\mathbf{C}}^d$, $f : U\in{\mathcal{U}}(d^2)\mapsto \left\|\frac{2}{d^2}U(P_E-P_{E^{\perp}})U^\dagger\right\|_{{\mathbf{\underline{M}}}}$ is a $\frac{8}{d}$-Lipschitz function. Indeed, by the same arguments as in the proof of \ref{eq:concentration},
\begin{eqnarray*}
f(U_1)-f(U_2) & \leq & \frac{2}{d^2}\left(\|U_1P_EU_1^\dagger-U_2P_EU_2^\dagger\|_{{\mathbf{\underline{M}}}} +\|U_1P_{E^{\perp}}U_1^\dagger-U_2P_{E^{\perp}}U_2^\dagger\|_{{\mathbf{\underline{M}}}} \right)\\
& \leq & \frac{2}{d} \left(\|U_1P_EU_1^\dagger-U_2P_EU_2^\dagger\|_2 +\|U_1P_{E^{\perp}}U_1^\dagger-U_2P_{E^{\perp}}U_2^\dagger\|_2 \right)\\
& \leq & \frac{4}{d} \left(\|U_1P_E-U_2P_E\|_2 +\|U_1P_{E^{\perp}}-U_2P_{E^{\perp}}\|_2 \right)\\
& \leq & \frac{8}{d} \|U_1-U_2\|_2.
\end{eqnarray*}
And any $L$-Lipschitz function $g : {\mathcal{U}}(n)\rightarrow{\mathbf{R}}$ satisfies the concentration estimate (see the Appendix in \cite{MM})
\[ \forall\ t>0,\ {\mathbf{P}}( | g - \E g| > t ) \leq 2 \exp(-cn t^2/L^2),\]
$c$ being a universal constant.

The function $f$ thus satisfies ${\mathbf{P}}(|f-\E f|>t)\leq 2\exp(-cd^4t^2)$. So the concentration estimate \ref{eq:concentration} in Proposition \ref{proposition:Delta-vs-difference-of-states} is in fact still true (and actually even stronger) for the random states under consideration.

What is more, the results from Lemma \ref{lemma:global-spectra} remain valid too because we here even have the equalities
\[ \left\|\frac{2}{d^2}U(P_E-P_{E^{\perp}})U^\dagger\right\|_1=2,\ \ \left\|\frac{2}{d^2}U(P_E-P_{E^{\perp}})U^\dagger\right\|_2=\frac{2}{d},\ \ \left\|\frac{2}{d^2}U(P_E-P_{E^{\perp}})U^\dagger\right\|_{\infty}=\frac{2}{d^2}. \]
So since $\frac{2}{d^2}U(P_E-P_{E^{\perp}})U^\dagger$ has the same distribution as $V\diag\left(\spec\left(\frac{2}{d^2}U(P_E-P_{E^{\perp}})U^\dagger\right)\right)V^\dagger$ for $V\in{\mathcal{U}}(d^2)$, one may apply Lemma \ref{lemma:comparison-of-norms} to conclude that the expectation estimate \ref{eq:expectation} in Proposition \ref{proposition:Delta-vs-difference-of-states} is in fact still true too for the random states under consideration.
}\end{proof}

In words, Theorem \ref{theorem:data-hiding} stipulates the following. Picking a subspace $E$ at random from the set of $\frac{d^2}{2}$-dimensional subspaces of ${\mathbf{C}}^d\otimes{\mathbf{C}}^d$, and then considering the states $\rho=\frac{P_E}{d^2/2}$ and $\sigma=\frac{P_{E^{\perp}}}{d^2/2}$, one gets examples of states which are perfectly distinguishable by some global measurement and which are with high probability data-hiding for separable measurements but not data-hiding for PPT measurements.

\begin{remark} Let us come back on the example of the symmetric state $\varsigma$ and the antisymmetric state $\alpha$ on ${\mathbf{C}}^d\otimes{\mathbf{C}}^d$. They satisfy (see e.g. \cite{DVLT2})
\begin{equation}
\label{eq:sym-anti}
\|\varsigma-\alpha\|_{{\mathbf{\underline{SEP}}}} = \|\varsigma-\alpha\|_{{\mathbf{\underline{PPT}}}} = \frac{4}{d+1} = \frac{2}{d+1}\|\varsigma-\alpha\|_{{\mathbf{\underline{ALL}}}}.
\end{equation}
They are consequently ``exceptional'' data hiding states for two reasons. First, as mentioned before, because they are equally PPT and SEP data hiding. And second because they are ``more'' data hiding than generic states: their SEP norm is of order $\frac{1}{d}\ll\frac{1}{\sqrt{d}}$, hence almost reaching the known lower-bound valid for any states $\rho,\sigma$ on ${\mathbf{C}}^d\otimes{\mathbf{C}}^d$ (see e.g. \cite{MWW}) namely $\|\rho-\sigma\|_{{\mathbf{\underline{SEP}}}}\geq \frac{2}{d}\|\rho-\sigma\|_{{\mathbf{\underline{ALL}}}}$.
\end{remark}

\subsection{Multipartite vs bipartite data hiding}

In Theorem \ref{theorem:vrad-w-PPT-SEP}, we focused on the bipartite case $\cH=({\mathbf{C}}^d)^{\otimes 2}$ for the sake of clarity. However, generalizations to the
general $k$-partite case $\cH=({\mathbf{C}}^d)^{\otimes k}$ are quite straightforward, at least in the situation where the high-dimensional composite system of interest is made of a
``small'' number of ``large'' subsystems (i.e. $k$ is fixed and $d$ tends to infinity).

Let us denote by ${\mathbf{\underline{PPT}}}_{d,k}$ and ${\mathbf{\underline{SEP}}}_{d,k}$ the sets of respectively $k$-PPT and $k$-separable POVMs on $({\mathbf{C}}^d)^{\otimes k}$.
On the one hand, an iteration of the
Milman--Pajor inequality (Corollary \ref{corollary:Milman-Pajor}) leads to the estimate
\[ c^{2^k}d^{k/2} \leq {\mathrm{vrad}}(K_{{\mathbf{\underline{PPT}}}_{d,k}})\leq w(K_{{\mathbf{\underline{PPT}}}_{d,k}})\leq Cd^{k/2}, \]
for some constants $c,C$ depending neither on $k$ nor on $d$.

On the other hand, the generalization of Theorem \ref{theorem:separable-states} to the set ${\mathcal{S}}_{d,k}$ of $k$-separable states on $({\mathbf{C}}^d)^{\otimes k}$
is known, namely (see \cite{AS})
\[ \frac{c^{k}}{d^{k-1/2}}\leq {\mathrm{vrad}}({\mathcal{S}}_{d,k})\leq w({\mathcal{S}}_{d,k})\leq C\frac{\sqrt{k\log k}}{d^{k-1/2}}, \]
and implies that
\[ c^{k}d^{1/2}\leq {\mathrm{vrad}}(K_{{\mathbf{\underline{SEP}}}_{d,k}})\leq w(K_{{\mathbf{\underline{SEP}}}_{d,k}})\leq C\sqrt{k\log k}d^{1/2}, \]
for some constants $c,C$ depending neither on $k$ nor on $d$.

A multipartite analogue of Theorem \ref{theorem:typical-states} can then be derived, following the exact same lines of proof.

\begin{theorem}
\label{th:multiparty}
There exist constants $c_k,C_k$ such that the following holds. Given a dimension $d$, let $\rho$ and $\sigma$ be random states, independent and uniformly distributed on the set of states on $({\mathbf{C}}^d)^{\otimes k}$. Then, with high probability,
\[ c_k \leq \| \rho - \sigma \|_{{\mathbf{\underline{PPT}}}_{d,k}} \leq \| \rho - \sigma \|_{{\mathbf{\underline{ALL}}}} \leq C_k, \]
\[ \frac{c_k}{\sqrt{d^{k-1}}} \leq \| \rho - \sigma \|_{{\mathbf{\underline{SEP}}}_{d,k}} \leq \frac{C_k}{\sqrt{d^{k-1}}}. \]
\end{theorem}

This means that, forgetting about the dependence on $k$ and only focusing on the one on $d$, for typical states $\rho,\sigma$
on $({\mathbf{C}}^d)^{\otimes k}$,
$\|\rho-\sigma\|_{{\mathbf{\underline{PPT}}}_{d,k}}$ is of order $1$, like $\|\rho-\sigma\|_{{\mathbf{\underline{ALL}}}}$, while $\|\rho-\sigma\|_{{\mathbf{\underline{SEP}}}_{d,k}}$
is of order $1/\sqrt{d^{k-1}}$.

\medskip

In this multipartite setting, another quite natural question is the one of finding states that local observers can poorly distinguish if they remain alone but that they can distinguish substantially better though by gathering into any possible two groups. This type of problem was especially studied in \cite{EW}. Here is another result in that direction.

Define $\mathbf{\underline{bi-SEP}}_{d,k}$ as the set of POVMs on $({\mathbf{C}}^d)^{\otimes k}$ which are biseparable across any bipartition of $({\mathbf{C}}^d)^{\otimes k}$. It may then be shown that for random states $\rho,\sigma$, independent and uniformly distributed on the set of states on $({\mathbf{C}}^d)^{\otimes k}$, with high probability,
$\| \rho - \sigma \|_{\mathbf{\underline{bi-SEP}}_{d,k}} \simeq d^{-k/4}$ (whereas $\| \rho - \sigma \|_{{\mathbf{\underline{SEP}}}_{d,k}} \simeq d^{-(k-1)/2}$ by Theorem \ref{th:multiparty}). This means that on $({\mathbf{C}}^d)^{\otimes k}$, with $k>2$ fixed, restricting to POVMs which are biseparable across every bipartition
is roughly the same as restricting to POVMs which are biseparable across one bipartition, whereas imposing $k$-separability is a much tougher constraint that implies a dimensional loss in the distinguishing ability.

\begin{remark} This result might not be as strong as one could hope for. It only shows that $\|\cdot\|_{\mathbf{\underline{bi-SEP}}_{d,k}}$ typically vanishes
slower than $\|\cdot\|_{{\mathbf{\underline{SEP}}}_{d,k}}$ when the local dimension $d$ grows, but it does not provide examples of states
$\rho,\sigma$ on $({\mathbf{C}}^d)^{\otimes k}$ for which $\|\rho-\sigma\|_{\mathbf{\underline{bi-SEP}}_{d,k}}$ would be of order $1$ while
$\|\rho-\sigma\|_{{\mathbf{\underline{SEP}}}_{d,k}}$ would tend to zero.
\end{remark}

\section{Miscellaneous remarks and questions}

\subsection{Complexity of the different classes of POVMs on a bipartite system}

Having at hand the estimates on the mean width of $K_{{\mathbf{\underline{SEP}}}}$ (or $K_{{\mathbf{\underline{LOCC}}}}$) and $K_{{\mathbf{\underline{PPT}}}}$ provided by Theorem \ref{theorem:vrad-w-PPT-SEP}, one may follow the exact same lines as in the proof of Theorem \ref{theorem:approximation-of-ALL} to identify the number of POVMs needed to approximate the corresponding locally restricted classes of POVMs. It is thus possible to show that on ${\mathbf{C}}^d\otimes{\mathbf{C}}^d$, $\exp(\Theta(d^4))$ different POVMs are necessary and sufficient to approximate the class ${\mathbf{\underline{PPT}}}$. For the class ${\mathbf{\underline{SEP}}}$ (or ${\mathbf{\underline{LOCC}}}$), we lack a complete answer since the same arguments show that the minimal number of POVMs is between $\exp(\Omega(d^3))$ and $\exp(O(d^4))$.

Let us make another comment on that topic. Theorem \ref{theorem:typical-states} tells us, amongst other, that the class of PPT POVMs is, in some sense, a quite good approximation of the class of all POVMs. One may therefore wonder if there would be a way, when trying to approximate the class of all POVMs by a finite sub-family, to impose that all POVMs in it are PPT. However, since the approximation we are looking for is one in terms of distinguishability norms (i.e. one that is valid for \textit{any} pair of states to be discriminated), this possibility is ruled out by the fact that the gap between $\|\cdot\|_{{\mathbf{\underline{PPT}}}}$ and $\|\cdot\|_{{\mathbf{\underline{ALL}}}}$ is unbounded (i.e. that there exist pair of states, such as e.g. the Werner states, which are poorly distinguished by any PPT POVM).

\subsection{What is the typical performance of the class ${\mathbf{\underline{LO}}}$?}
\label{sec:LO-LOCC-SEP}

While Theorem \ref{theorem:LO-vs-LOCC} shows that the gap between the classes ${\mathbf{\underline{LO}}}$ and ${\mathbf{\underline{LOCC}}}$ may be unbounded,
we do not know if this situation is typical or not. Asking whether norms are comparable in a typical direction is more or less equivalent to asking
whether the ratio $\vrad(K_{{\mathbf{\underline{LOCC}}}}) / \vrad(K_{{\mathbf{\underline{LO}}}})$ is bounded as the dimension increases.

\subsection{Can the gap between ${\mathbf{\underline{LOCC}^\rightarrow}}/{\mathbf{\underline{LOCC}}}/{\mathbf{\underline{SEP}}}$ be unbounded?}
Or conversely, does there exist an absolute
constant $c$ such that the norm inequalities $\|\cdot\|_{{\mathbf{\underline{LOCC}^\rightarrow}}} \geq c \|\cdot\|_{{\mathbf{\underline{LOCC}}}}$ and/or
$\|\cdot\|_{{\mathbf{\underline{LOCC}}}} \geq c \|\cdot\|_{{\mathbf{\underline{SEP}}}}$ hold for any dimension?

\subsection{Locally restricted measurements on a multipartite quantum system}
There are at least two ways for a multipartite system such as $({\mathbf{C}}^d)^{\otimes k}$
to be of high dimension: either with $k$ fixed and $d$ large
(few large subsystems) or $k$ large and $d$ fixed (many small subsystems). Theorem \ref{th:multiparty} tells us what is the
typical discriminating power of $k$-PPT and $k$-separable POVMs, but in the first setting only. The extension to the case of many small subsystems
seems a challenging problem.

\nonumsection{Acknowledgements}
\noindent

This research was supported by the ANR project OSQPI  ANR-11-BS01-0008.

\nonumsection{References}

\appendix{. Classical convex geometry}

\subappendix{Some vocabulary}
\label{ap:convex-geometry}

We work in the Euclidean space ${\mathbf{R}}^n$, where we denote by $\|\cdot\|_2$ the Euclidean norm. We denote by $\vol_n(\cdot)$ or simply $\vol(\cdot)$ the
$n$-dimensional Lebesgue measure. A {\em convex body} $K \subset {\mathbf{R}}^n$ is a convex compact set with non-empty interior.
A convex body $K$ is {\em symmetric} if $K=-K$.  The {\em gauge}  associated to
a convex body $K$ is the function $\|\cdot\|_K$ defined for $x \in {\mathbf{R}}^n$ by
$\|x\|_K := \inf \{ t \geq 0 \ : \ x \in tK \}$. This is a norm if and only if $K$ is symmetric.

If $K \subset {\mathbf{R}}^n$ is a convex body with origin in its interior,
the {\em polar} of $K$ is the convex body $K^\circ$ defined as
\[ K^\circ := \{ y \in {\mathbf{R}}^n \ : \ \langle x,y \rangle \leq 1 \ \ \hbox{\rm for all  } \ x \in K \} .\]
In the symmetric case, the norms $\|\cdot\|_K$ and $\|\cdot\|_{K^\circ}$ are dual to each other.

If $u$ is a vector from the unit sphere $S^{n-1}$, the {\em support
function} of $K$ in the direction $u$ is
\[ h_K(u):=\max_{x\in K}
\langle x, u\rangle = \|u\|_{K^\circ}. \]
Note that $h_K(u)$ is the
distance from the origin to the hyperplane tangent to $K$ in the
direction $u$.

Two global invariants associated to a convex body $K \subset {\mathbf{R}}^n$, the {\em volume radius} and the {\em mean width}, play an important role in our proofs.

\begin{definition}
The {\em volume
radius} of a convex body $K \subset {\mathbf{R}}^n$ is defined as
\[ \vrad (K) := \left( \frac{\vol K}{\vol B_2^n} \right)^{1/n} ,\]
where $B_2^n$ denotes the unit Euclidean ball of ${\mathbf{R}}^n$.

In words, $\vrad(K)$ is the radius of the Euclidean ball with same volume as $K$.
\end{definition}

\begin{definition}
The {\em mean width} of a subset $K \subset {\mathbf{R}}^n$ is defined as
\[ w (K):=\int _{S^{n-1}} \max_{x\in X} \langle x,u \rangle \,{\mathrm{d}}\sigma(u) ,\]
where ${\mathrm{d}}\sigma(u)$ is the normalized spherical
measure on the unit Euclidean sphere $S^{n-1}$ of ${\mathbf{R}}^n$.
If $K$ is a convex body, we have
\[ w (K):=\int _{S^{n-1}}
h_K(u)\,{\mathrm{d}}\sigma(u)=\int _{S^{n-1}} \|u\|_{K^\circ} \, {\mathrm{d}}\sigma(u).
\]
\end{definition}

The inequality below (see, e.g., \cite{Pisier}) is a
fundamental result which compares the volume radius and the mean
width.

\begin{theorem}[Urysohn inequality]
\label{theorem:urysohn}
For any convex body $K \subset {\mathbf{R}}^n$, we have
\[ \vrad(K) \leq w(K) .\]
\end{theorem}

It is convenient to compute the mean width using Gaussian rather than spherical integration. Let $G$ be a standard Gaussian vector in ${\mathbf{R}}^n$,
i.e. such that its
coordinates, in any orthonormal basis, are independent with a $N(0,1)$ distribution. Denoting $\gamma_n = \E \|G\|_2 \sim \sqrt{n}$, we have,
for any compact set
$K \subset {\mathbf{R}}^n$,
\[ w_G(K) := \E \max_{x \in K} \langle G,x \rangle = \gamma_n w(K) .\]
The Gaussian mean width is usually easier to compute. For example, it allows to compute the mean width of a segment: if $u \in S^{n-1}$ is a
unit vector, then
\[ \alpha_n := w (\conv \{ \pm u\}) = \frac{1}{\gamma_n} \sqrt{\frac{2}{\pi}} \sim \sqrt{\frac{2}{\pi n}}. \]
It also shows how to control the mean width of a projection. Let $K \subset {\mathbf{R}}^n$ be a compact set, and $E \subset {\mathbf{R}}^n$ be a $k$-dimensional subspace.
Denoting $P_E$ the orthogonal projection onto $E$, we have
$w_G(P_E K) \leq w_G(K)$, and therefore
\begin{equation} \label{eq:width-projection}
w(K \cap E) \leq w(P_E K) \leq \frac{\gamma_n}{\gamma_k} w(K).
\end{equation}

We also need the following lemma which is an incarnation of the familiar ``union bound'' and appears for example as formula (3.6) in \cite{LT}
(under the equivalent formulation via suprema of Gaussian processes).

\begin{lemma}[Bounding the mean width of a union] \label{lemma:mean-width-union-bound}
Let $K_1,\dots,K_N$ be convex symmetric sets in ${\mathbf{R}}^n$ such that $K_i \subset \lambda B_2^n$ for every index $1\leq i\leq N$
(where $B_2^n$ denotes the unit Euclidean ball of ${\mathbf{R}}^n$). Then
\[ w \left( \conv \left( \bigcup_{i=1}^N K_i \right) \right) \leq C \left( \max_{1 \leq i \leq N} w(K_i) +
\lambda \sqrt{\frac{\log N}{n}}  \right),\]
where $C$ is an absolute constant.
\end{lemma}

\subappendix{Some volume inequalities}
\label{ap:volume}

We use repeatedly the following result, established in \cite{MP}, Corollary 3.

\begin{theorem}[Milman--Pajor inequality]
\label{theorem:Milman-Pajor}
Let $K,L$ be convex bodies in ${\mathbf{R}}^n$ with the same center of mass. Then
\[ {\mathrm{vrad}}(K\cap L){\mathrm{vrad}}(K-L)\geq{\mathrm{vrad}}(K){\mathrm{vrad}}(L).\]
\end{theorem}

Choosing $K=-L$ in Theorem \ref{theorem:Milman-Pajor} yields the following corollary.

\begin{corollary}
\label{corollary:Milman-Pajor}
If $K$ is a convex body in ${\mathbf{R}}^n$ with center of mass at the origin, then
\[{\mathrm{vrad}}(K\cap -K)\geq\frac{1}{2}{\mathrm{vrad}}(K),\]
and more generally for any orthogonal transformation $\theta$,
\[{\mathrm{vrad}}(K\cap\theta(K))\geq\frac{1}{2}\frac{{\mathrm{vrad}}(K)^2}{w(K)}.\]
\end{corollary}

We typically use Corollary \ref{corollary:Milman-Pajor} in the following way: if $K$ is a convex body with center of
mass at the origin which satisfies a ``reverse'' Urysohn inequality,
i.e. ${\mathrm{vrad}}(K) \geq \alpha w(K)$ for some constant $\alpha$, we conclude that the volume radius of $K \cap \theta(K)$
is comparable to the volume radius of $K$.

Another volume inequality which is useful to us is the Rogers--Shepard inequality (see \cite{RS}).

\begin{theorem}[Rogers--Shephard inequality] \label{theorem:rogers-shephard}
Let $u$ be a unit vector in ${\mathbf{R}}^n$, $h>0$ and consider the affine hyperplane
\[ H = \{ x \in {\mathbf{R}}^n \ : \ \scalar{x}{u} = h \} .\]
Let $K$ be a convex body inside $H$ and $L = \conv (K,-K)$. Then,
\[  2h \vol_{n-1}(K) \leq \vol_n(L) \leq 2h \vol_{n-1}(K) \frac{2^{n-1}}{n}. \]
%with equality in the second inequality when $K$ is a simplex.
Consequently,
\begin{equation} \label{eq:RS} {\mathrm{vrad}}(L) \simeq h^{1/n} {\mathrm{vrad}}(K)^{1-1/n}. \end{equation}
\end{theorem}

We can infer from equation \ref{eq:RS} that for sets $K$ with ``reasonable'' volume (which will be the case of all sets we consider)
${\mathrm{vrad}}(K)$ and ${\mathrm{vrad}}(L)$ are comparable.

\subappendix{Volume estimates for Schatten classes and related bodies}
\label{ap:standard}

We gather estimates on mean width and volume radius of ``standard'' sets, which are used in our proofs. We use the following notation
for the unit balls associated to Schatten norms
\[ S^d_1 = \{ A \in {\mathcal{H}}({\mathbf{C}}^d) \ : \ \|A\|_1 \leq 1 \}, \]
\[ S^d_{\infty} = \{ A \in {\mathcal{H}}({\mathbf{C}}^d) \ : \ \|A\|_\infty \leq 1 \} = [ -{\mathrm{Id}} , {\mathrm{Id}}]. \]

Moreover, given symmetric convex bodies $K \subset {\mathbf{R}}^n$ and $K' \subset {\mathbf{R}}^{n'}$, their projective tensor product is defined as
\[ K \hat{\otimes} K' = \conv \{ x \otimes x' \ : \ x \in K, x' \in K' \} \subset {\mathbf{R}}^n \otimes {\mathbf{R}}^{n'} \]

\begin{theorem}
\label{theorem:operator-norm}
We have
\[{\mathrm{vrad}}(S_{\infty}^d) \simeq  w(S_{\infty}^d) \simeq \sqrt{d}.\]
\[{\mathrm{vrad}}(S_1^d) \simeq  w(S_1^d) \simeq \frac{1}{\sqrt{d}}.\]
\end{theorem}

\begin{proof}{
The estimates on the mean width follow from the semicircle law. Indeed,
the standard Gaussian vector in the space of self-adjoint operators on ${\mathbf{C}}^d$ is exactly a GUE matrix $G$ (see \cite{AGZ}), and therefore
\[ w_G(S_{\infty}^d) = \E \|G\|_1 = d^{3/2} \int_{-2}^2 |x| \frac{\sqrt{4-x^2}}{2\pi} \,{\mathrm{d}}x = d^{3/2} \frac{8}{3\pi} ,\]
\[ w_G(S_1^d) = \E \|G\|_{\infty} = (2+o(1)) \sqrt{d}. \]
Hence, $w(S_{\infty}^d) = \gamma_{d^2}^{-1} w_G(S_{\infty}^d) \sim \frac{8}{3\pi} \sqrt{d}$ and $w(S_1^d)=\gamma_{d^2}^{-1} w_G(S_1^d) \sim \frac{2}{\sqrt{d}}$.

Since $S_1^d$ and $S_{\infty}^d$ are polar to each other, the Santal\'o inequality (see \cite{Santalo}) yields
\[ 1 \leq \vrad(S_{\infty}^d) \vrad(S_1^d).\]
If we then use the Urysohn inequality, we obtain
\[ 1 \leq w(S_{\infty}^d)w(S_1^d) \leq \frac{8\sqrt{d}}{3\pi} \frac{2}{\sqrt{d}} \simeq 1 ,\]
and therefore all these inequalities are sharp up to a multiplicative constant.
}\end{proof}

We also need volume estimates on projective tensor products of Schatten spaces.

\begin{theorem} \label{theorem:volume-S1-Sinfini}
We have the following estimates
\[ {\mathrm{vrad}}(S_1^d \hat{\otimes} S_{\infty}^d) \simeq w(S_1^d \hat{\otimes} S_{\infty}^d) \simeq \frac{1}{\sqrt{d}}.\]
\end{theorem}

A very similar proof shows that the estimates of Theorem \ref{theorem:volume-S1-Sinfini} are also valid when we consider the full complex Schatten
classes, without the self-adjoint constraint. The question of estimating the volume radius of projective
tensor product of Schatten classes has been considered in \cite{DM}, where the question is answered (in a general setting) only up to a factor $\log d$.

\begin{proof}{
An upper bound on the mean width can be obtained by a discretization argument, which we only sketch since we will only use the lower bound.
There is a polytope $P$ with $\exp(Cd)$ vertices such that $S_1^d \subset P \subset 2S_1^d$, and a polytope $Q$ with $\exp(Cd^2)$ vertices
such that $S_{\infty}^d \subset Q \subset 2S_{\infty}^d$. The polytope $P \hat{\otimes} Q$ satisfies
\[ S_1^d \hat{\otimes} S_{\infty}^d \subset P \hat{\otimes} Q \subset 4 S_1^d \hat{\otimes} S_{\infty}^d .\]
The polytope $P \hat{\otimes} Q$ is the convex hull of $\exp(C'd^2)$ points with Hilbert--Schmidt norm at most $4\sqrt{d}$. Using standard bounds
for mean width of polytopes (see e.g. \cite{AS}) gives the desired estimate $w(S_1^d \hat{\otimes} S_{\infty}^d) \,\preceq\, 1/\sqrt{d}$.

We now give a lower bound on the volume radius.
We denote by $B_1^n \subset {\mathbf{R}}^n$ the unit ball of the space $\ell_1^n$. We have the following formula.

\begin{lemma} \label{lemma:volume-ell_1-projective}
Let $m,n$ be integers and $K \subset {\mathbf{R}}^m$ be a symmetric convex body. Then
\[ \vol ( B_1^n \hat{\otimes} K ) = \frac{(m!)^n}{(mn)!} \vol(K)^n .\]
Consequently,
\[ \vrad (B_1^n \hat{\otimes} K ) \simeq \frac{1}{\sqrt{n}} \vrad(K) .\]
\end{lemma}

\begin{proof}{
If $(e_1,\dots,e_n)$ denotes the canonical basis of ${\mathbf{R}}^n$, we have, for any $x_1,\dots,x_n \in {\mathbf{R}}^m$
\[ \left\| \sum_{i=1}^n e_i \otimes x_i \right\|_{B_1^n \hat{\otimes} K} = \sum_{i=1}^n \|x_i\|_K .\]
So Lemma \ref{lemma:volume-ell_1-projective} follows easily from the formula below, valid for any integer $p$ and any symmetric convex body $L\subset{\mathbf{R}}^p$,
\begin{equation}\label{eq:volume-convex-body} \vol(L) = \frac{1}{p!} \int_{{\mathbf{R}}^p} \exp ( - \|x\|_L) \, {\mathrm{d}}x .\end{equation}
Equation \ref{eq:volume-convex-body} itself may be obtained by the following chain of equalities
\begin{eqnarray*} \int_{{\mathbf{R}}^p} e^{- \|x\|_L} \, {\mathrm{d}}x & = & \int_{{\mathbf{R}}^p}\int_{\|x\|_L}^{+\infty} e^{-t}  \, {\mathrm{d}}t  \, {\mathrm{d}}x\\
 & = & \int_{0}^{+\infty}\int_{\{\|x\|_L<t\}} e^{-t}  \, {\mathrm{d}}x  \, {\mathrm{d}}t\\
 & = & \int_{0}^{+\infty} e^{-t} \vol(tL)  \, {\mathrm{d}}t \\
 & = & \vol(L)p!, \end{eqnarray*}
the last equality being because $\int_{0}^{+\infty} t^pe^{-t}\, {\mathrm{d}}t = p!$.
}\end{proof}

Denote by $\{\ket{j}\}_{1 \leq j \leq d}$ an orthonormal basis of ${\mathbf{C}}^d$. The family
\[ \left\{ \ketbra{j}{j} \right\}_{1 \leq j \leq d } \cup \left\{ \frac{1}{\sqrt{2}} ( \ketbra{j}{k} + \ketbra{k}{j} )\right\}_{ 1 \leq j < k \leq d }
 \cup \left\{ \frac{i}{\sqrt{2}} ( \ketbra{j}{k} - \ketbra{k}{j} ) \right\}_{ 1 \leq j < k \leq d } \]
is an orthonormal basis of $\cH({\mathbf{C}}^d)$ whose elements live in $\sqrt{2} S_1^d$. It follows that
\[ \vrad(S_1^d \hat{\otimes} S_{\infty}^d) \geq \frac{1}{\sqrt{2}} \vrad(B_1^{d^2} \hat{\otimes} S_{\infty}^d) \,\succeq\, \frac{1}{d} \vrad(S_{\infty}^d) ,\]
the last estimate being a consequence of Lemma \ref{lemma:volume-ell_1-projective}.

Using Theorem \ref{theorem:operator-norm} one may thus conclude that
$\vrad(S_1^d \hat{\otimes} S_{\infty}^d) \,\succeq\, 1/\sqrt{d}$.
}\end{proof}

We also need a result on the volume radius and the mean width of the set of separable states, which is taken from \cite{AS}.

\begin{theorem}
\label{theorem:separable-states}
In ${\mathcal{H}}({\mathbf{C}}^d\otimes{\mathbf{C}}^d)$, denoting by ${\mathcal{S}}$ the set of separable states, we have
\[d^{-3/2}\simeq {\mathrm{vrad}}({\mathcal{S}})\leq w({\mathcal{S}})\simeq d^{-3/2}.\]
\end{theorem}


\noindent

\begin{thebibliography}{000}

\bibitem{MWW} \textbf{W. Matthews, S. Wehner, A. Winter}, ``Distinguishability of quantum states under restricted families of measurements with an application
to data hiding'', Comm. Math. Phys. 291(3) (2009); arXiv:0810.2327[quant-ph].

\bibitem{Holevo} \textbf{A.S. Holevo}, ``Statistical decision theory for quantum systems'',
     J. Mult. Anal. 3, 337--394 (1973).

\bibitem{Helstrom} \textbf{C.W. Helstrom}, \textit{Quantum detection and estimation theory},
     Academic Press, New York, 1976.

\bibitem{AL} \textbf{G. Aubrun, C. Lancien}, ``Zonoids and sparsification of quantum measurements'', preprint; arXiv:1309.6003

\bibitem{LW} \textbf{C. Lancien, A. Winter}, ``Distinguishing multi-partite states by local measurements'', Commun. Math. Phys. 323, 555--573 (2013); arXiv[quant-ph]:1206.2884.

\bibitem{CH} \textbf{E. Chitambar, M-H. Hsieh}, ``Asymptotic state discrimination and a strict hierarchy in distinguishability norms''; arXiv:1311.1536[quant-ph].

\bibitem{DVLT2} \textbf{D.P. DiVincenzo, D. Leung, B.M. Terhal}, ``Quantum Data Hiding'', IEEE Trans. Inf Theory 48(3), 580-599 (2002); arXiv:quant-ph/0103098.

\bibitem{HLSW} \textbf{P. Hayden, D. Leung, P. Shor, A. Winter}, `` Randomizing quantum states: Constructions and applications'', Commun. Math. Phys. 250(2), 371--391 (2004);
arXiv:quant-ph/0307104.

\bibitem{Pisier} \textbf{G. Pisier}, \textit{The Volume of Convex Bodies and Banach Spaces Geometry}, Cambridge Tracts in Mathematics Volume 94, Cambridge University
Press, Cambridge, 1989.

\bibitem{DVHLST} \textbf{D.P. DiVincenzo, M. Horodecki, D. Leung, J. Smolin, B.M. Terhal}, ``Locking classical correlation in quantum states'', Phys. Rev. Lett. 92.067902 (2004); arXiv:quant-ph/0303088.

\bibitem{DFHL} \textbf{F. Dupuis, J. Florjanczyk, P. Hayden, D. Leung}, ``Locking classical information'', Proc. R. Soc. A, Vol. 469, No. 2159 (2013); arXiv:1011.1612[quant-ph].

\bibitem{FHS} \textbf{O. Fawzi, P. Hayden, P. Sen}, ``From low-distortion norm embeddings to explicit uncertainty relations and efficient information locking'', Journal of the ACM, Vol. 60, No. 6, Article 44 (2013); arXiv:1010.3007[quant-ph].

\bibitem{ZS} \textbf{K. \.{Z}yczkowski, H-J. Sommers}, ``Induced measures in the space of mixed quantum states'', J. Phys. A. 34, 7111--7124 (2001); arXiv:quant-ph/0012101.

\bibitem{Levy} \textbf{P. L\'{e}vy}, \textit{Probl\`{e}mes concrets d'analyse fonctionnelle} (French), 2nd ed. Gauthier-Villars, Paris, 1951.

\bibitem{Bhatia} \textbf{R. Bhatia}, \textit{Matrix analysis}, Graduate Texts in Mathematics, Vol. 169, Springer-Verlag, New-York, 1997.

\bibitem{DVLT1} \textbf{D.P. DiVincenzo, D. Leung, B.M. Terhal}, ``Hiding Bits in Bell States'', Phys. Rev. Lett. 86(25), 5807--5810 (2001); arXiv:quant-ph/0011042.

\bibitem{EW} \textbf{T. Eggeling, R.F. Werner}, ``Hiding classical data in multi-partite quantum states'', Phys. Rev. Lett. 89.097905 (2002); arXiv:quant-ph/0203004.

\bibitem{MM} \textbf{E. Meckes, M. Meckes}, ``Spectral measures of powers of random matrices'', Electron. Commun. Probab. 18.78, 1--13 (2013); arXiv:1210.2681[math.PR].

\bibitem{LT} \textbf{M. Ledoux, M. Talagrand} \textit{Probability in Banach Spaces: isoperimetry and processes}, Ergebnisse der Mathematik und ihrer Grenzgebiete, Vol. 23, Springer-Verlag,
Berlin Heidelberg, 1991.

\bibitem{MP} \textbf{V.D. Milman, A. Pajor}, ``Entropy and asymptotic geometry of non-symmetric convex bodies'', Advances in Math. 152, 314--335 (2000).

\bibitem{RS} \textbf{C.A. Rogers, G.C.Shephard}, ``Convex bodies associated with a given convex body'' {\em J. London Math. Soc.} 33, 270--281 (1958).

\bibitem{AGZ} \textbf{G.W. Anderson, A. Guionnet, O. Zeitouni}, \textit{An Introduction to Random Matrices}, Cambridge Studies in Advanced Mathematics, Vol. 118, Cambridge University
Press, Cambridge, 2010.

\bibitem{Santalo} \textbf{L. Santal\'o}, ``An affine invariant for convex bodies of $n$-dimensional space'' (Spanish), Portugaliae Math. 8, 155--161 (1949).

\bibitem{DM} \textbf{A. Defant, C. Michels}, ``Norms of tensor product identities.'' Note di Matematica 25.1, 129--166 (2006).

\bibitem{AS} \textbf{G. Aubrun, S.J. Szarek}, ``Tensor product of convex sets and the volume of separable states on N qudits'', Phys. Rev. A. 73 (2006); arXiv:quant-ph/0503221.
















\end{thebibliography}
\end{document}